%% file: ms.tex
\pdfoutput=1
\documentclass[sigconf,anonymous=false]{acmart}

\setcopyright{rightsretained}
\acmConference{preprint}{2019}{arXiv.org} 
\copyrightyear{2019}

\input{defns}

\begin{document}
\title{Invariant, Viability and Discriminating Kernel
  Under-Approximation via Zonotope Scaling}
%\titlenote{Happy to hear alternative suggestions for the title.}
%\subtitle{Extended Abstract}
%\subtitlenote{The full version of the author's guide is available as \texttt{acmart.pdf} document}

\author{Ian M. Mitchell}
\orcid{0000-0001-7053-441X}
\affiliation{%
  \department{Department of Computer Science}
  \institution{The University of British Columbia}
%  \streetaddress{2366 Main Mall}
  \city{Vancouver}
  \state{British Columbia}
%  \postcode{V6T 1Z4}
  \country{Canada}}
\email{mitchell@cs.ubc.ca}

\author{Jacob Budzis}
%\orcid{0000-0001-7053-441X}
\affiliation{%
  \department{Engineering Physics Program}
  \institution{The University of British Columbia}
%  \streetaddress{??}
  \city{Vancouver}
  \state{British Columbia}
%  \postcode{V6T 1Z4}
  \country{Canada}}
\email{jrbudzis@gmail.com}

\author{Andriy Bolyachevets}
%\authornote{This is how to add footnotes to individual authors.}
\affiliation{%
  \department{Department of Computer Science}
  \institution{The University of British Columbia}
%  \streetaddress{2366 Main Mall}
  \city{Vancouver}
  \state{British Columbia}
%  \postcode{V6T 1Z4}
  \country{Canada}}
\email{andriy.bolyachevets@alumni.ubc.ca}


\begin{abstract}
  \input{abstract}
\end{abstract}

\ignore{

}


\maketitle

\input{introduction}

\input{related}

\input{preliminaries}

\input{invariance}

\input{viability}

\input{discriminating}

%\input{example}

\input{conclusion}

\begin{acks}
This work was supported by \grantsponsor{NSERC}{National Science and
  Engineering Research Council of
  Canada}{http://www.nserc-crsng.gc.ca/index_eng.asp} (NSERC)
Undergraduate Student Research Awards and Discovery Grant
\#\grantnum{NSERC}{298211}.
\end{acks}

\bibliographystyle{ACM-Reference-Format}
\bibliography{hscc-zonotope-viability} 

\end{document}

%% file: defns.tex
\usepackage[T1]{fontenc}

\usepackage{todonotes}

\usepackage{booktabs} % For formal tables

\hypersetup{pdfpagemode=UseNone}

\usepackage{paralist}
\usepackage{caption}
\usepackage{subcaption}

\graphicspath{{Figures/}}
\usepackage{ifpdf}
\ifpdf\DeclareGraphicsExtensions{.pdf,.png,.jpg}\fi

\newcommand{\matlab}{\textsc{Matlab}}

\newcommand{\R}{\mathbb{R}}

\newcommand{\defined}{\triangleq}

\newcommand{\discOp}{\mathsf{Disc}}
\newcommand{\discSet}[2]{\discOp\left({#1},{#2}\right)}
\newcommand{\invOp}{\mathsf{Inv}}
\newcommand{\invSet}[2]{\invOp\left({#1},{#2}\right)}
\newcommand{\viabOp}{\mathsf{Viab}}
\newcommand{\viabSet}[2]{\viabOp\left({#1},{#2}\right)}
\newcommand{\reachOp}{\mathsf{R}}
\newcommand{\reachSet}[2]{\reachOp\left({#1},{#2}\right)}

\newcommand{\centvec}[1]{\centvecelem{}{#1}}
\newcommand{\centvecelem}[2]{{c}_{#1}\left({#2}\right)}
\newcommand{\genvec}[2]{{g}_{#1}\left({#2}\right)}
\newcommand{\genmat}[1]{\mat{G}\left({#1}\right)}
\newcommand{\countvecOp}{p}
\newcommand{\countvec}[1]{\countvecOp\left({#1}\right)}
\newcommand{\countvecI}{\countvec{\set I}}
\newcommand{\bigphi}{\Phi}
\newcommand{\controlDisturbSet}{\set F}

\newcommand{\element}[2]{({#1})_{#2}}

\newcommand{\ones}[1]{\mathbf{1}_{#1}}
\newcommand{\zeros}[1]{\mathbf{0}_{#1}}
\newcommand{\mat}[1]{\mathrm{#1}}
\newcommand{\set}[1]{\mathcal{#1}}
\newcommand{\com}[1]{\texttt{#1}}

\newcommand{\sig}{(\cdot)}

\newcommand{\bma}{\begin{bmatrix}}
\newcommand{\ema}{\end{bmatrix}}

\long\def\ignore#1{}
\newcommand{\refeq}[1]{(\ref{#1})}

\newcommand{\todoIM}[1]{\todo[inline,color=green,author=IM]{#1}}

\usepackage{xcolor}
\theoremstyle{acmdefinition}
\newtheorem{remark}[theorem]{Remark}

%% file: abstract.tex
Scalable safety verification of continuous state dynamic systems has been
demonstrated through both reachability and viability analyses using parametric
set representations; however, these two analyses are not interchangable in
practice for such parametric representations.  In this paper we consider
viability analysis for discrete time affine dynamic systems with adversarial
inputs.  Given a set of state and input constraints, and treating the inputs
in best-case and/or worst-case fashion, we construct invariant, viable and
discriminating sets, which must therefore under-approximate the invariant,
viable and discriminating kernels respectively.  The sets are constructed by
scaling zonotopes represented in center-generator form.  The scale factors are
found through efficient convex optimizations.  The results are demonstrated on
two toy examples and a six dimensional nonlinear longitudinal model of a
quadrotor.
%\footnote{This is an abstract footnote}

%% file: introduction.tex
\section{Introduction} \label{s:intro}

Reachability analysis is a rigorous alternative to sampling-based
verification of dynamic systems, and at least for linear (or affine)
systems there have been recent demonstrations of techniques capable of
handling thousands of continuous state space
dimensions~\cite{bakduggirala17a,Bogomolov+18}.  Reachable sets and
tubes---or more typically over-approximations of them to ensure
soundness---are an effective tool for demonstrating safety: If the
forward / backward reach set or tube does not intersect the unsafe /
initial set respectively then the system is safe.  Any input or
parameter uncertainty is typically treated in a worst-case fashion to
make the reach set or tube larger and hence the system less likely to
be judged safe.

In this paper we adopt the alternative approach to proving safety
advocated in viability theory~\cite{ABSP11} but more commonly framed
as various versions of invariant sets: Find the set of states from
which trajectories of the system are guaranteed to satisfy a specified
safety constraint.  A critical feature of viable or controlled
invariant sets is that a control input may be chosen in a best-case
fashion to keep the trajectories safe.  Robust versions of these sets
(aka discriminating sets) also allow an adversarial disturbance input
which is treated in a worst-case fashion to drive the trajectories out
of the safety constraints.  In every case we must under-approximate
the results to ensure soundness of the safety analysis.

While viability analysis can be reduced to reachability analysis and
vice versa in theory, the parametric representations which can handle
high dimensional systems do not support such reductions; for example,
the reductions require set complements but the parametric
representations are usually restricted to convex sets whose
complements are non-convex.

The need to develop scalable algorithms for viability / invariance in
addition to those for reachability is therefore practical: The former
require under-approximation while the latter over-approximation, and
some analyses are more naturally amenable to parametric
representations in one formulation or the other, but rarely both.
Although we do not have space to explore it in this paper, an example
of an application which naturally fits into the viability framework is
testing at run-time whether an exogenous input signal---such as might
arise from a human-in-the-loop control---will maintain safety; for
example, see~\cite{MYLT16}.

The focus of this paper is therefore development of more scalable
formulations for (robust) invariance / viability based analysis of
affine continuous state dynamic systems.  Scalability is achieved by
framing the calculations as convex optimizations to find efficient
parametric representations in the form of zonotopes.  The specific
contributions of this paper are to:
\begin{compactitem}
  \item Show how a finite horizon invariance kernel of an affine
    system with disturbance input can be underapproximated with a
    zonotope via a linear program.
  \item Extend the formulation to allow control inputs and thereby
    underapproximate finite horizon viability and discriminating
    kernels with zonotopes via a convex program.
  \item Demonstrate the use of the discriminating kernel to compute a
    larger robust controlled invariant set for a moderate dimensional
    nonlinear quadrotor model than was achieved using an ellipsoidal
    representation in~\cite{MYLT16}.
\end{compactitem}

%% file: related.tex
\section{Related Work} \label{s:related}

Reachability and viability have been applied to a broad variety of
different dynamic systems resulting in a huge range of different
algorithms.  We focus here on parametric approaches for linear or
affine dynamics.  Ellipsoid and support vector parametric
representations of viability constructs were explored
in~\cite{MKMOD13,MK15}.  Zonotopic representations for reachability
were introduced in~\cite{girard05} and have since been extensively
explored; for example~\cite{GGM06,AK11,AF16}.  Our work was inspired,
however, by the papers~\cite{SA17a,SA17b} and in
particular~\cite{SA17c}, which utilizes a convex optimization to
select zonotope generator weights to construct a control scheme that
will drive a set of initial states into the smallest possible set of
final states.  Also similar to this work is~\cite{HREA16} in which the
authors seek a linear feedback control input which will maximize the
size of the backward reachable set and arrive at a bilinear matrix
inequality based on containment of one zonotope within another.
Significantly, unlike most other work these papers treat the input in
a best-case fashion.  The difference with the work below lies in the
reachability vs viability formulation and the fact that our approach
constructs a set-valued viable feedback control from a convex
optimization.

%% file: preliminaries.tex
%-------------------------------------------------------------
\section{Preliminaries} \label{s:prelim}

For a matrix $\mat M$, let $|\mat M|$ denote the elementwise absolute
value and $\element{\mat M}{i,j}$ denote the element in row $i$ and
column $j$; therefore, $\element{\mat M^s}{i,j}$ denotes the element
in row $i$ and column $j$ of the matrix power $\mat M^s$.  Also define
$\ones{d}$ and $\zeros{d}$ to be the vectors in $\R^d$ of all ones and
all zeros respectively.
%, and $\powerset{\set S}$ be the set of all subsets of $\set S$.

%-------------------------------------------------------------
\subsection{System Dynamics}

Consider a discrete time dynamic system for times $t = 0, 1, 2,
\ldots$:
\begin{equation} \label{e:linear-dynamics}
  x(t+1) = \mat A x(t) + \mat B u(t) + \mat C v(t) + w
\end{equation}
where
\begin{compactitem}
  \item The \emph{state} is $x \in \R^{d_x}$.
  \item The \emph{control input} is $u \in \set U \subset \R^{d_u}$.  The
    control input constraint $\set U$ is an interval hull (or
    hyperrectangle) defined by the elementwise inequalities
    \begin{equation} \label{e:control-constraints}
      \set U = \{ u \in \R^{d_u} \mid \underline u \leq u \leq \overline u \}.
    \end{equation}
  \item The \emph{disturbance input} is $v \in \set V \subset
    \R^{d_v}$.  The disturbance input constraint $\set V$ is a
    zonotope (see section~\ref{s:zonotope}).
  \item The \emph{drift} $w \in \R^{d_x}$ is constant.  We note that
    $w \neq 0$ can alternatively be treated by offsetting the center
    of the zonotope $\set V$, but we will carry $w$ through separately
    so as to support a drift term for disturbance-free systems.
\end{compactitem}
%We assume that dynamics matrix $\mat A$ is non-singular.  This
%condition will hold, for example, if~\refeq{e:linear-dynamics} is the
%result of a time discretization of a continuous time system.

Given an initial condition $x(0)$ and input signals $u\sig$ and
$v\sig$, the solution of~\refeq{e:linear-dynamics} is given by
\begin{equation} \label{e:linear-soln}
  x(t) = \mat A^t x(0)
    + \sum_{s=0}^{t-1} \mat A^{t-1-s} \left( \mat B u(s) + \mat C v(s) + w \right).
\end{equation}

%-------------------------------------------------------------
\subsection{Sets of Interest for Safety Verification}

We will formulate our safety verification problem as keeping the system state within a constraint set $\set X$.  For reasons of notational simplicity, we will assume in the rest of the paper that $\set X$ is an interval hull or box constraint of the form
\begin{equation} \label{e:state-constraints}
    \set X = \{ x \in \R^{d_x} 
        \mid \underline x \leq x \leq \overline x \}.
\end{equation}
It is possible to relax this assumption; see
section~\ref{s:discussion} for further comments.

We seek to approximate various subsets of the constraint set.  The most general is a finite horizon \emph{discriminating} or \emph{robust controlled invariant} set
\begin{equation} \label{e:discriminating-defn}
  \discSet{[0,T]}{\set X} \defined
      \left\{ x(0) \in \set X  \,\left|\,
      \begin{gathered}
        \exists u\sig, \forall v\sig, \forall t \in [ 0, T ],\\
         x(t) \in \set X 
      \end{gathered}
      \right. \right\},
\end{equation}
where $u(t) \in \set U$ and $v(t) \in \set V$ for all $t = 0, \ldots,
T$.  Note that the control input $u(t)$ tries to keep the system state
within the constraint set $\set X$, while the disturbance input $v(t)$
is treated in a worst-case or adversarial fashion and tries to drive
the system state outside the constraint set.

We will also consider two special cases of the discriminating set.  For systems which lack a control input, we seek a finite horizon \emph{invariant} set
\begin{equation} \label{e:invariant-defn}
  \invSet{[0,T]}{\set X} \defined
      \left\{ x(0) \in \set X  \,\left|\,
      \begin{gathered}
        \forall v\sig, \forall t \in [ 0, T ],\\
         x(t) \in \set X 
      \end{gathered}
      \right. \right\},
\end{equation}
while for systems which lack a disturbance input we seek a finite horizon \emph{viable} or \emph{controlled invariant} set
\begin{equation} \label{e:viable-defn}
  \viabSet{[0,T]}{\set X} \defined
      \left\{ x(0) \in \set X  \,\left|\,
      \begin{gathered}
        \exists u\sig, \forall t \in [ 0, T ],\\
         x(t) \in \set X 
      \end{gathered}
      \right. \right\}.
\end{equation}

Our algorithm for computing these previous sets will often make use of the forward \emph{reach} set of some specified set $\set S$.
\begin{equation} \label{e:reach-defn}
  \reachSet{t}{\set S} \defined
      \left\{ x(t) \in \R^{d_x} \,\left|\,
        \exists w\sig, x(0) \in \set S 
      \right. \right\},
\end{equation}
where the choice of input signal $w\sig = u\sig$ or $w\sig = v\sig$
should be clear from context.
Unlike the discriminating, invariant and viable sets, the reach set is defined at a single time rather than over an interval, and the set $\set S$ is an initial condition rather than a constraint.

%-------------------------------------------------------------
\subsection{Set Representation: Zonotopes} \label{s:zonotope}

We will use \emph{zonotopes} as our parametric representation of invariant, viable and/or discriminating sets. A zonotope $\set S \subseteq \R^d$ is a polytope defined by a center $\centvec{\set S} \in \R^d$ and a finite number of generators $\genvec{i}{\set S} \in \R^d$ for $i = 1, \ldots, \countvec{\set S}$:
\begin{equation} \label{e:zonotope-defn-components}
  \set S = \left\{ \left. \centvec{\set S} 
                   + \sum_{i=1}^{\countvec{\set S}} 
                   \lambda_i \genvec{i}{\set S}
                   \;\right|\; {-1} \leq \lambda_i \leq +1 \right\}.
\end{equation}
In most cases it is more convenient to work with the center-generator
tuple (or \emph{G-representation} for a zonotope $\set S$ rather
than~\refeq{e:zonotope-defn-components}
\[
  \begin{aligned}
    \set S &= \langle \centvec{\set S} \mid \genvec{1}{\set S},
        \genvec{2}{\set S}, \ldots, 
        \genvec{\countvec{\set S}}{\set S} \rangle, \\
      &= \langle \centvec{\set S} \mid \genmat{\set S} \rangle,
  \end{aligned}
\]
where the generator matrix $\genmat{\set S}$ is formed by horizontal concatenation of the generator vectors
\[
    \genmat{\set S} = \bma \genvec{1}{\set S} & \genvec{2}{\set S} & \cdots 
         & \genvec{\countvec{\set S}}{\set S} \ema 
       \in \R^{d \times \countvec{\set S}}.
\]
When it is necessary to refer to individual elements of a generator
matrix or vector, we will use the notation $\genvec{j,i}{\set S}$ to
specify the element in the $j^{th}$ row and $i^{th}$ column of matrix
$\genmat{\set S}$, or equivalently the $j^{th}$ element of vector
$\genvec{i}{\set S}$.

With the generator matrix notation, we can
write~\refeq{e:zonotope-defn-components} more compactly as
\begin{equation} \label{e:zonotope-defn}
  \set S = \left\{ \centvec{\set S} + \genmat{\set S} \lambda
                     \;\left|\; {-\ones{\countvec{\set S}}} \leq \lambda
                            \leq +\ones{\countvec{\set S}} \right. \right\}.
\end{equation}
For a vector $x \in \set S$, define $\lambda(x, \set S)$ such that
\begin{equation} \label{e:lambda-defn}
    x = \centvec{\set S} + \genmat{\set S} \lambda(x, \set S)
    \text{ and } 
    {-\ones{\countvec{\set S}}} \leq \lambda(x, \set S)
      \leq +\ones{\countvec{\set S}} 
\end{equation}
and note that by~\refeq{e:zonotope-defn} $\lambda(x, \set S)$ exists
but is not necessarily unique.

The coordinate bounds for a zonotope (or equivalently the interval
hull containing that zonotope) are easily computed; for example,
see~\cite{GGM06,AK11}.  In particular, for $x \in \set S$,
\begin{equation} \label{e:zonotope-box-constraints-elementwise}
  \centvecelem{j}{\set S}
    - \sum_{i=0}^{\countvec{\set S}} \left|\genvec{j,i}{\set S}\right|
  \leq x_j \leq
  \centvecelem{j}{\set S}
    + \sum_{i=0}^{\countvec{\set S}} \left|\genvec{j,i}{\set S}\right|,
\end{equation}
for all $j = 1, \ldots, d$, which we can write in compact form as
\begin{equation} \label{e:zonotope-box-constraints}
  \centvec{\set S} - |\genmat{\set S}| \ones{\countvec{\set S}}
  \leq x \leq
  \centvec{\set S} + |\genmat{\set S}| \ones{\countvec{\set S}}.
\end{equation}

\begin{lemma} \label{t:zonotope-box-containment}
  Consider an interval hull
  \[
    \set B = \{ b \in \R^d 
        \mid \underline b \leq b \leq \overline b \}.
  \]
  For zonotope $\set S \subset \R^d$, the containment $\set S
  \subseteq \set B$ is equivalent to the constraints
  \begin{equation} \label{e:zonotope-box-containment-elementwise}
    \begin{aligned}
      \centvecelem{j}{\set S}
          - \sum_{i=0}^{\countvec{\set S}} \left|\genvec{j,i}{\set S}\right|
        &\geq \underline b_j, \\
      \centvecelem{j}{\set S}
          + \sum_{i=0}^{\countvec{\set S}} \left|\genvec{j,i}{\set S}\right|,
        &\leq \overline b_j
    \end{aligned}
  \end{equation}
  for all $j = 1, \ldots, d$.  More compactly,
  \begin{equation} \label{e:zonotope-box-containment}
    \begin{aligned}
      \centvec{\set S} - |\genmat{\set S}| \ones{\countvec{\set S}}
        &\geq \underline b, \\
      \centvec{\set S} + |\genmat{\set S}| \ones{\countvec{\set S}}
        &\geq \overline b.
    \end{aligned}
  \end{equation}
\end{lemma}

\begin{proof}
  A straightforward consequence of the
  bounds~\refeq{e:zonotope-box-constraints-elementwise}
  and~\refeq{e:zonotope-box-constraints}.
\end{proof}

The class of zonotopes is closed under linear
transformation~\cite{girard05}, and the effect of a linear transform
on a zonotope is easily implemented using the center-generator tuple
representation
\[
    \mat M \set S = \langle \mat M \centvec{\set S} \mid 
      \mat M \genmat{\set S} \rangle, \\
\]
where $\set S$ is a zonotope in dimension $d$ and $\mat M \in \R^{m
  \times d}$ is the matrix representing the linear transformation.

\begin{proposition} \label{t:reach-set-evolution-no-control}
  For systems without control inputs but with disturbance inputs
  constrained by the zonotope
  \[
    \set V = \langle \centvec{\set V} | \genmat{\set V} \rangle,
  \]
  the center, generator count and generator matrix of the exact reach
  set $\reachSet{t}{\set S}$ are
  \begin{equation} \label{e:reach-set-zonotope-no-control}
    \begin{aligned}
      \centvec{\reachSet{t}{\set S}} &= \mat A^t \centvec{\set S} 
        + \sum_{s=0}^{t-1} \mat A^{t-1-s}
          \left( \mat C \centvec{\set V} + w \right), \\
      \countvec{\reachSet{t}{\set S}} &= \countvec{\set S}
        + t \countvec{\set V}, \\    
      \genmat{\reachSet{t}{\set S}} &= 
        \bma \mat A^t \genmat{\set S} 
             & \mat A^{t-1} \mat C \genmat{\set V}
             %& \mat A^{t-2} \mat C \genmat{\set V}
             & \cdots
             & \mat C \genmat{\set V} \ema
   \end{aligned}
  \end{equation}
\end{proposition}

\begin{proof}
    See~\cite{GGM06,althoff10a}
\end{proof}

\ignore{        
% This was the generator vector version.  Too complicated and replaced with the generator matrix version above.
    \countvec{\reachSet{t}{\set S}} &= \countvec{\set S} + t \countvec{\set V}, \\
    \genvec{i}{\reachSet{t}{\set S}} &= \mat A^t \genvec{i}{\set S}, 
      \text{ for } i = 1, \ldots, \countvec{\set S}, \\
    \genvec{i + \countvec{\set S}}{\reachSet{t}{\set S}} &= 
      \mat A^{t-1} C \genvec{i}{\set V},
      \text{ for } i = 1, \ldots, \countvec{\set V}, \\
    \genvec{i + \countvec{\set S} + \countvec{\set V}}{\reachSet{t}{\set S}} &= 
      \mat A^{t-2} C \genvec{i}{\set V},
      \text{ for } i = 1, \ldots, \countvec{\set V}, \\
    &\vdots\\
    \genvec{i + \countvec{\set S} + (t-1) \countvec{\set V}}{\reachSet{t}{\set S}} &= 
      C \genvec{i}{\set V},
      \text{ for } i = 1, \ldots, \countvec{\set V}.
}

%-------------------------------------------------------------
\subsection{Computational Environment}

The examples in the paper were implemented in \matlab\ (R2018a) using
CVX~\cite{cvx, gb08} (version~2.1) with the SDPT3 solver (version~4.0)
for the convex optimizations and CORA~\cite{althoff15} (2018~release)
for zonotope visualization.  Timings were taken on a Lenovo ThinkPad
X1 Yoga 1st Signature Edition laptop with an Intel Core i7-6500U (dual
core) at 2.5~GHz with 8~GB memory running Windows~10 Pro
(version~1803).

%% file: invariance.tex
%-------------------------------------------------------------
\section{Computing Invariant Sets} \label{s:invariance}

As defined in~\refeq{e:invariant-defn}, invariant sets are computed for systems without control input, so in this section we focus on the case where $\set U = \emptyset$.  In order to minimize notational complexity, we first sketch the algorithm for systems with no disturbance input before proceeding to the more general case.

%-------------------------------------------------------------
\subsection{With No Inputs} \label{s:no-inputs}

In this section we will derive conditions under which a parameterized
zonotope $\set I$ is invariant with respect to the state
constraints~\refeq{e:state-constraints}.  Define
\begin{equation} \label{e:invariant-zonotope-defn}
  \begin{aligned}  
    \set I &= \langle \alpha \mid \gamma_1 \genvec{1}{\set I}, 
        \gamma_2 \genvec{2}{\set I}, \ldots, 
        \gamma_{\countvec{\set I}} \genvec{\countvec{\set I}}{\set I} \rangle, \\
      &=\langle \alpha \mid \genmat{\set I} \Gamma \rangle,
  \end{aligned}
\end{equation}
where
\[
  \Gamma = \bma \gamma_1 & 0 & \cdots & 0 \\
                0 & \gamma_2 & \cdots & 0 \\
                \vdots & \vdots & \ddots & \vdots \\
                0 & 0 & \cdots & \gamma_{\countvec{\set I}} \ema
         \in \R^{\countvec{\set I} \times \countvec{\set I}}
\]
is the diagonal matrix with vector $\gamma$ along its diagonal.  In
this parameterization, $\genmat{\set I}$ is a specified set of
generators, but the center vector $\alpha \in \R^{d_x}$ and
\emph{generator scaling vector} $\gamma \in \R^{\countvec{\set I}}$
with $\gamma \geq 0$ are free parameters.  Each element of $\gamma$ is
associated with a generator of $\set I$ and can be thought of
intuitively as the ``width'' of $\set I$ in that generator direction.

\begin{proposition}
  Assuming a system with no disturbance or control input, the reach
  set for an initial state space zonotope $\set I$ can be represented
  by a center, generator count and generator matrix given by
  \begin{equation} \label{e:zonotope-no-input-evolution}
    \begin{aligned}
      \centvec{\reachSet{t}{\set I}} &= \mat A^t \alpha
        + \sum_{s=0}^{t-1} \mat A^{t-1-s} w \\
      \countvec{\reachSet{t}{\set I}} &= \countvec{\set I} \\
      \genmat{\reachSet{t}{\set I}} &= A^t \genmat{\set I} \Gamma,
    \end{aligned}
  \end{equation}
  where the first equation can be written out elementwise as
  \begin{equation} \label{e:zonotope-no-input-evolution-center}
    \centvecelem{j}{\reachSet{t}{\set I}} = 
      \sum_{k=1}^{d_x} \left( \element{\mat A^t}{j,k} \alpha_k
      + \sum_{s=0}^{t-1} \element{\mat A^{t-1-s}}{j,k} w_k \right)
  \end{equation}
  and the last equation can be written out for the $i^{th}$
  generator elementwise as
  \begin{equation} \label{e:zonotope-no-input-evolution-elementwise}
      \genvec{j,i}{\reachSet{t}{\set I}}
        = \sum_{k=1}^{d_x} \element{\mat A^t}{j,k} \genvec{k,i}{\set I} \gamma_i.
  \end{equation}
\end{proposition}

\begin{proof}
  Straightforward substitution of $\set V = \emptyset$
  and~\refeq{e:invariant-zonotope-defn}
  into~\refeq{e:reach-set-zonotope-no-control}.
\end{proof}

Knowing how the reach set $\reachSet{t}{\set I}$ evolves allows us to
ensure that trajectories which start within $\set I$ stay within the
constraint set $\set X$.

\begin{proposition} \label{t:invariant-set-no-input-containment}
  Assuming a system with no disturbance or control input, $\set I
  \subseteq \invSet{[0,T]}{\set X}$ if
  \begin{equation} \label{e:invariant-no-input-box-containment-elementwise}
    \begin{aligned}
      \sum_{k=1}^{d_x} \left( \element{\mat A^t}{j,k} \alpha_k
        + \sum_{s=0}^{t-1} \element{\mat A^{t-1-s}}{j,k} w_k \right)
        - \sum_{i=0}^{\countvec{\set I}}
            \left| \sum_{k=1}^{d_x} \element{\mat A^t}{j,k}
                    \genvec{k,i}{\set I} \right| \gamma_i
      & \geq \underline x_j \\
      \sum_{k=1}^{d_x} \left( \element{\mat A^t}{j,k} \alpha_k
        + \sum_{s=0}^{t-1} \element{\mat A^{t-1-s}}{j,k} w_k \right)
        + \sum_{i=0}^{\countvec{\set I}}
            \left| \sum_{k=1}^{d_x} \element{\mat A^t}{j,k}
                    \genvec{k,i}{\set I} \right| \gamma_i
      & \leq \overline x_j
    \end{aligned}
  \end{equation}
  for all $j = 1, \ldots, d_x$ and $t = 0, \ldots, T$.  More
  compactly,
  \begin{equation} \label{e:invariant-no-input-box-containment}
    \begin{aligned}
       \mat A^t \alpha + \sum_{s=0}^{t-1} \mat A^{t-1-s} w
         - \left|\mat A^t \genmat{\set I}\right| \gamma
        &\geq \underline x, \\
       \mat A^t \alpha + \sum_{s=0}^{t-1} \mat A^{t-1-s} w
         + \left|\mat A^t \genmat{\set I}\right| \gamma
        &\leq \overline x. \\
    \end{aligned}
  \end{equation}
  for all $t = 0, \ldots, T$.
\end{proposition}

\begin{proof}
  To show that $\set I \subseteq \invSet{[0, T]}{\set X}$,
  by~\refeq{e:invariant-defn} we need to show for $x(0) \in \set I$
  that $x(t) \in \set X$.  If $x(0) \in \set I$ then $x(t) \in
  \reachSet{t}{\set I}$ by~\refeq{e:reach-defn}.
  Plugging~\refeq{e:zonotope-no-input-evolution-center}
  and~\refeq{e:zonotope-no-input-evolution-elementwise}
  into~\refeq{e:zonotope-box-containment-elementwise}
  yields~\refeq{e:invariant-no-input-box-containment-elementwise}, and
  so by~\refeq{e:state-constraints} and
  Lemma~\ref{t:zonotope-box-containment}, the
  constraint~\refeq{e:invariant-no-input-box-containment-elementwise}
  implies
  \begin{equation} \label{e:reach-no-input-containment}
    \reachSet{t}{\set I} \subseteq \set X 
      \text{ for all } t = 0, \ldots, T;
  \end{equation}
  consequently, $x(t) \in \set X$.  Note that we can move $\gamma_i$
  outside the absolute value because $\gamma_i \geq 0$.
  Rearranging~\refeq{e:invariant-no-input-box-containment-elementwise}
  and taking advantage of the fact $\Gamma \ones{\countvec{\set I}} =
  \gamma$ yields~\refeq{e:invariant-no-input-box-containment}.
\end{proof}

\ignore{
  Combine~\refeq{e:state-constraints},~\refeq{e:zonotope-box-constraints-elementwise},~\refeq{e:invariant-no-input-containment},~\refeq{e:zonotope-no-input-evolution}
  and~\refeq{e:zonotope-no-input-evolution-elementwise} we arrive at
  the constraints

  In subsequent
derivations, we will take advantage
of~\refeq{e:zonotope-box-constraints} to write the constraints for all
dimensions at a fixed time $t$ in more compact form as
}

\begin{remark}
  The constraints~\refeq{e:invariant-no-input-box-containment} are
  linear in $\alpha$ and $\gamma$.
\end{remark}

Ideally, we would then seek the set $\set I$ of maximum volume which satisfies~\refeq{e:invariant-no-input-box-containment}.  Although an analytic formula for zonotope volume exists~\cite{GK10}, it is combinatorially complex in the number of generators and hence we settle for the simpler heuristic of maximizing the sum of the elements of $\gamma$ (and thereby the sum of the ``widths'' of the generators).  Our algorithm can therefore be written as an optimization problem
\begin{equation} \label{e:invariant-optimization-no-input}
  \begin{aligned}
  \max_{\alpha, \gamma} \; &\ones{\countvec{\set I}}^T \gamma \\
  \text{such that } & \gamma \geq 0 \\
  \text{and } & \text{\refeq{e:invariant-no-input-box-containment} holds }
                \forall t \in 0, \ldots, T\\
  \end{aligned}
\end{equation}

\begin{remark}
  The optimization~\refeq{e:invariant-optimization-no-input} is a
  linear program with $d_x + \countvec{\set I}$ decision variables,
  $\countvec{\set I}$ non-negativity constraints, and $2 d_x (T+1)$
  general constraints
  from~\refeq{e:invariant-no-input-box-containment}.
\end{remark}

%-------------------------------------------------------------
\subsection{With Disturbance Inputs} \label{s:with-disturbance}

For systems with uncertainty in the form of a disturbance input $v \in
\set V \neq \emptyset$, we must include the effect of $\set V$ on
$\reachSet{t}{\set I}$ when we construct the constraints for our
optimization.

\begin{proposition} \label{t:invariant-set-disturb-containment}
  Assuming a system with no control input, $\set I \subseteq
  \invSet{[0,T]}{\set X}$ if
  \begin{equation} \label{e:invariant-box-containment}
    \begin{aligned}
      \left( \begin{gathered}
        \mat A^t \alpha + \sum_{s=0}^{t-1} \mat A^{t-1-s}
                          \left( \mat C \centvec{\set V} + w \right) \\
          - \left|\mat A^t \genmat{\set I}\right| \gamma
          - \sum_{s=0}^{t-1} \left(
               \left| \mat A^{t-1-s} \mat C \genmat{\set V} \right|
                           \ones{\countvec{\set V}} \right)
      \end{gathered} \right)
      &\geq \underline x, \\
      \left( \begin{gathered}
        \mat A^t \alpha + \sum_{s=0}^{t-1} \mat A^{t-1-s}
                          \left( \mat C \centvec{\set V} + w \right) \\
          + \left|\mat A^t \genmat{\set I}\right| \gamma
          + \sum_{s=0}^{t-1} \left(
               \left| \mat A^{t-1-s} \mat C \genmat{\set V} \right|
                           \ones{\countvec{\set V}} \right)
      \end{gathered} \right)
      &\leq \overline x. \\
    \end{aligned}
  \end{equation}
  for all $t = 0, \ldots, T$.
\end{proposition}

\begin{proof}
  Plugging~\refeq{e:reach-set-zonotope-no-control}
  into~\refeq{e:zonotope-box-containment} demonstrates
  that~\refeq{e:invariant-box-containment} implies $\reachSet{t}{\set
    I} \subseteq \set X$ when $\set V \neq \emptyset$.  The remainder
  of the proof follows that of
  Proposition~\ref{t:invariant-set-no-input-containment}.
\end{proof}

Note that the generators arising from the disturbance input appear in
the constraints but are not scaled.  The
constraints~\refeq{e:invariant-box-containment} are still linear in
$\alpha$ and $\gamma$, so we can compute an invariant set using the
linear program optimization~\refeq{e:invariant-optimization-no-input}
with~\refeq{e:invariant-box-containment} substituted
for~\refeq{e:invariant-no-input-box-containment}.

%-------------------------------------------------------------
\subsection{Example: Rotational Dynamics}
\label{s:invariant-example}

\begin{figure}
  \includegraphics[width=0.22\textwidth]{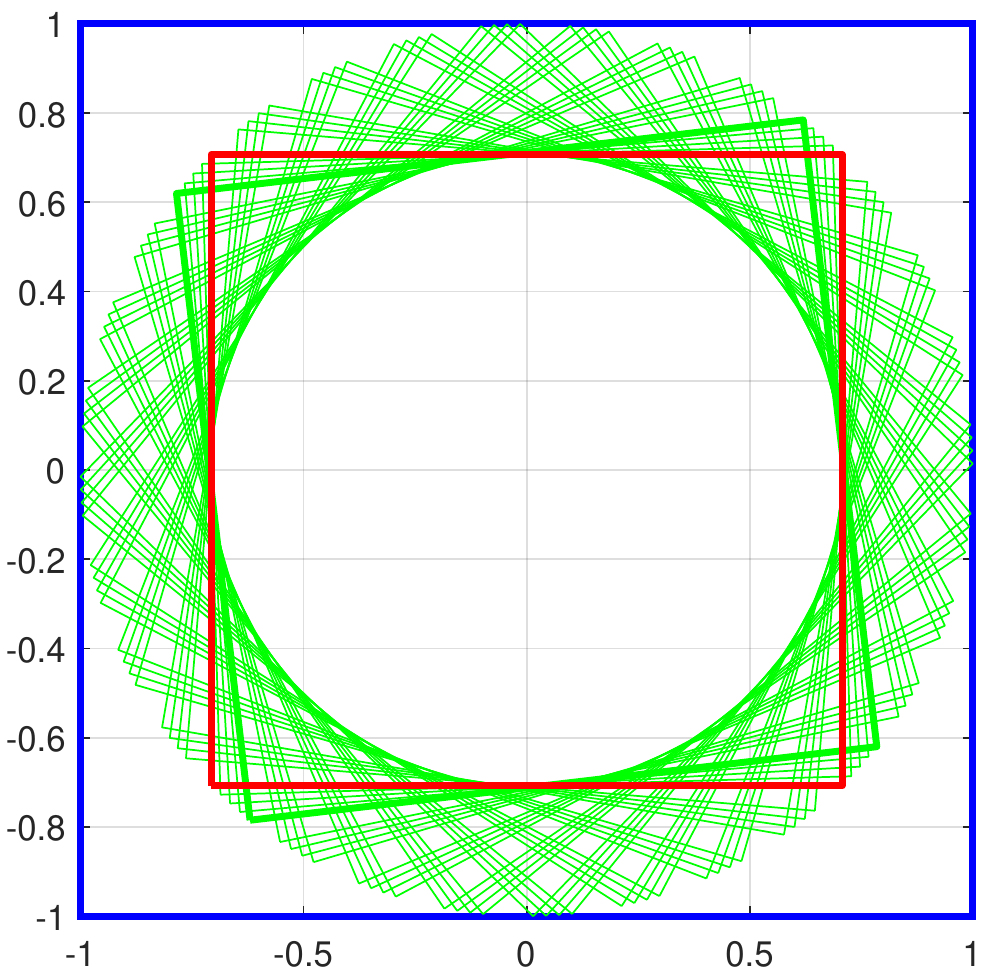}
  \includegraphics[width=0.22\textwidth]{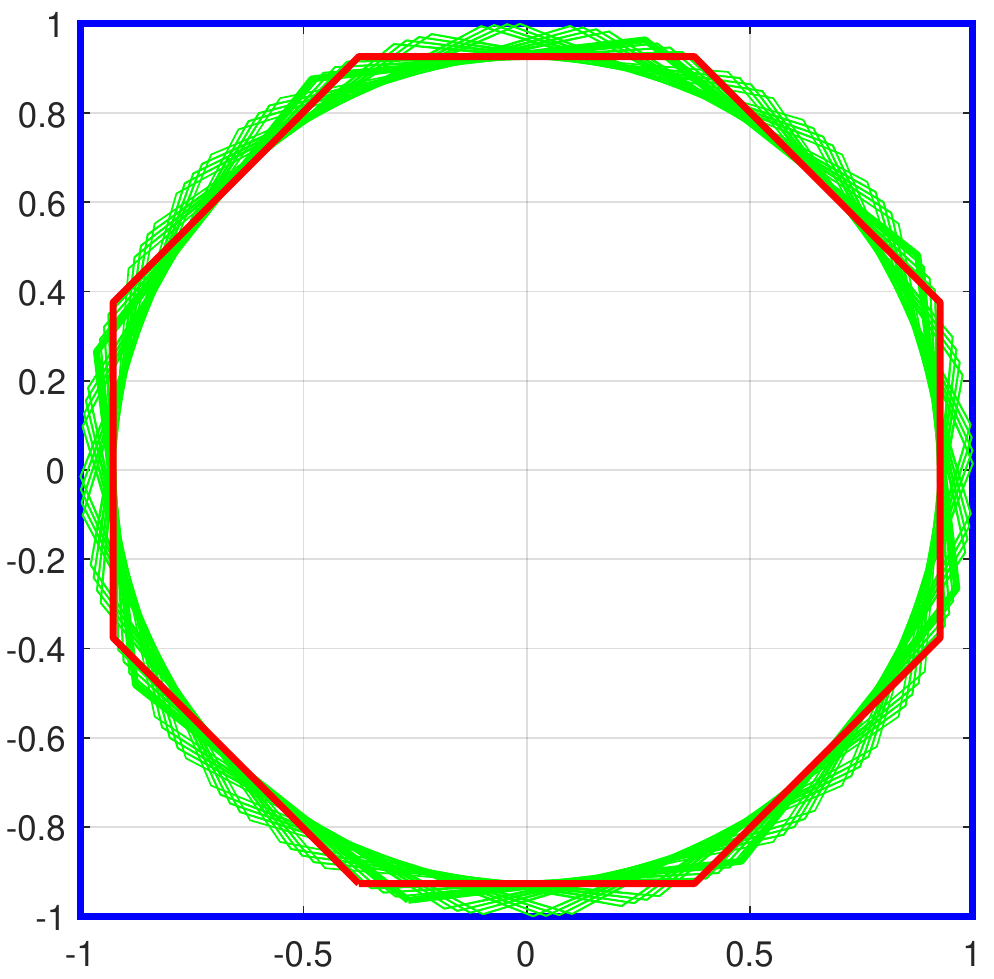} \\
  \includegraphics[width=0.22\textwidth]{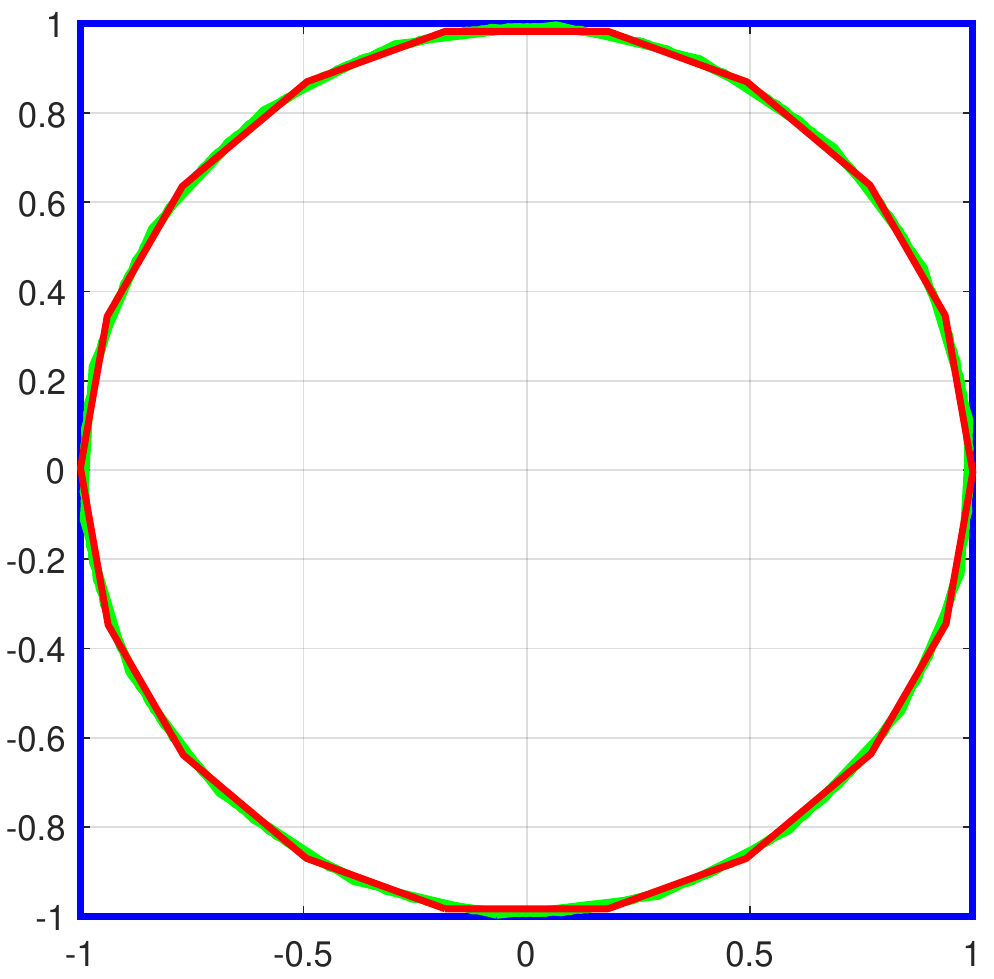}
  \includegraphics[width=0.22\textwidth]{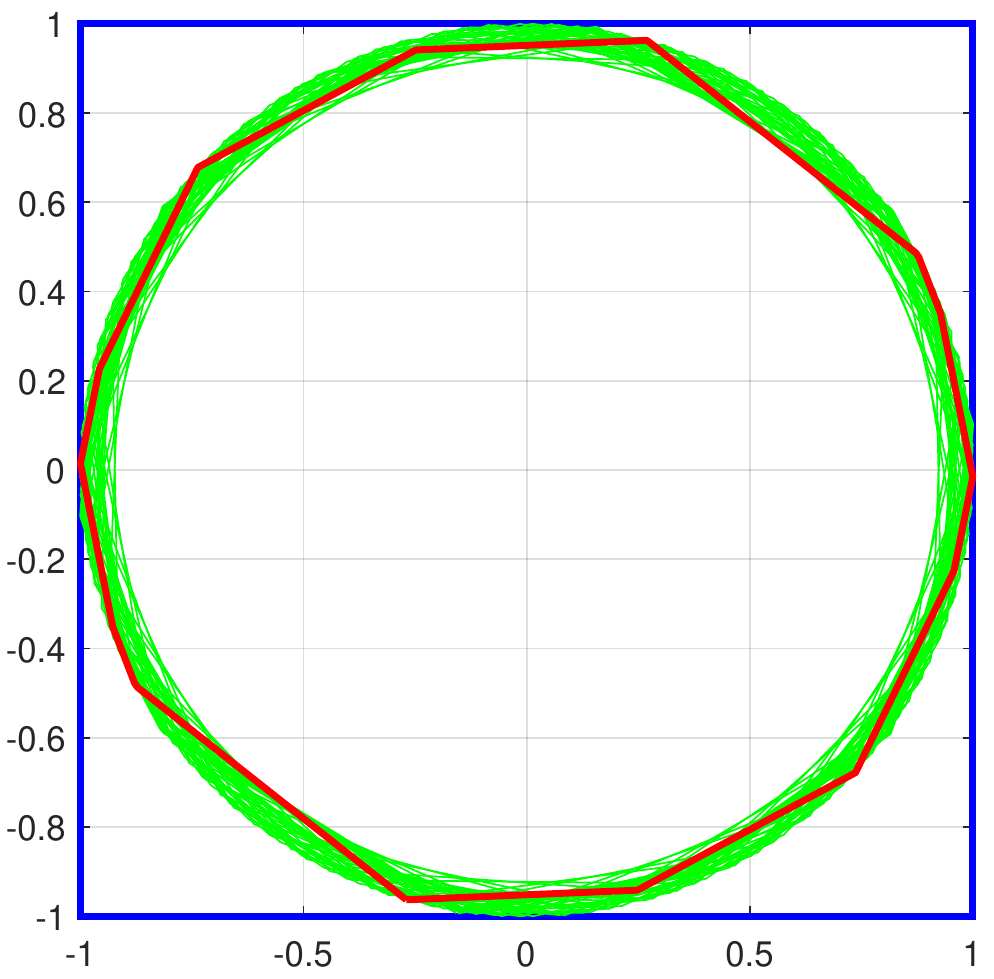}
  \caption{Computed invariant set $\set I$ (red thick line) for the
    rotation example with varying numbers of generators in $\set I$.
    Also shown are the constraint set $\set X$ (blue thick line) and
    $\reachSet{t}{\set I}$ (thin green lines) for $t = 1, \ldots T$.
    Top left: Two generators (the coordinate axes).  Top right: Four
    generators (coordinate axes plus diagonals).  Bottom left: Nine
    generators (equally spaced around top half circle).  Bottom right:
    Sixteen random generators (only seven have scaling $\gamma_i >
    0.01$).}
  \label{f:invariant-example}
\end{figure}

To demonstrate our algorithm and the effects of the choice of
$\genmat{\set I}$, we consider a system with rotational dynamics.
Starting from the continuous time system
\[
    \dot x = \bma 0 & -1 \\ +1 & 0 \ema x
\]
we use the matrix exponential with time step $0.2$ to construct the
discrete time system
\begin{equation} \label{e:toy-rotation-dynamics}
  x(t+1) = \bma +0.9801 & -0.1987 \\ +0.1987 & +0.9801 \ema x(t).
\end{equation}
We use the optimization~\refeq{e:invariant-optimization-no-input} with
$\underline x = {-\ones{2}}$, $\overline x = +\ones{2}$ and $T = 32$
to compute an invariant set after slightly more than one full
rotation.  The true invariance kernel in this case is the largest
circle contained in $\set X$.  Figure~\ref{f:invariant-example} shows
the computed invariant sets with different numbers of generators in
$\set I$.  Run times for the optimization were below 3~seconds for all
of these examples.

\begin{figure}
  \includegraphics[width=0.3\textwidth]{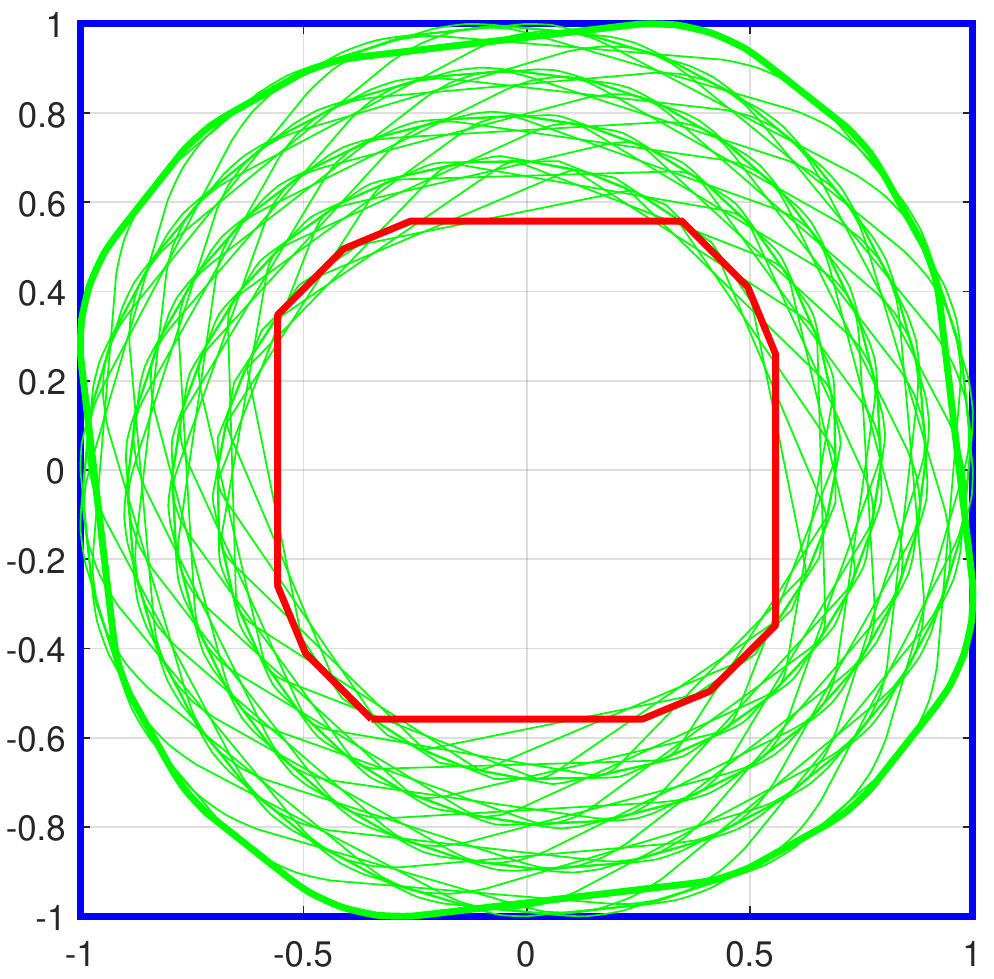}
  \caption{Computed invariant set for the rotation example with
    disturbance input using eight equally spaced generators for $\set
    I$.}
  \label{f:invariant-example-disturb}
\end{figure}

Figure~\ref{f:invariant-example-disturb} illustrates the results of
the calculation if we introduce a disturbance with
\[
    \mat C = \bma 1 & 0 \\ 0 & 1 \ema
    \qquad
    \set V = \left\langle
      \bma 0 \\ 0 \ema \left|
      \bma 0.05 & 0 \\ 0 & 0.05 \ema \right. \right\rangle
\]
into the rotational dynamics~\refeq{e:toy-rotation-dynamics}.

%% file: viability.tex
%-----------------------------------------------------------
\section{Computing Viable Sets} \label{s:viability}

Accommodating the disturbance input for invariant sets was
notationally complicated but conceptually straightforward: The
constraints on the reach sets of the initial set had to take into
account the effect of an input which could take on any value in $\set
V$ at any state at any time.  Achieving viability typically requires
that the control input be chosen based on the current state;
consequently, we need to tackle the control inputs in a different
manner.  To simplify the notation we work through the algorithm
without disturbance inputs in this section, and consequently compute
viable sets.  We also omit the drift term for now.

%-----------------------------------------------------------
\subsection{Augmenting Generators with Control}
\label{s:viability-augment-generators}

Define the function
\begin{equation} \label{e:augmented-zonotope-defn}
  \begin{aligned}  
    \set J&(\set I, \beta, \bigphi)\\
      &= \left\langle \bma \alpha \\ \beta \ema
         \;\left|\; \bma \gamma_1 \genvec{1}{\set I} \\ \phi_1 \ema, 
                    \bma \gamma_2 \genvec{2}{\set I} \\ \phi_2 \ema,
                    \ldots, 
                    \bma \gamma_{\countvecI}
                         \genvec{\countvecI}{\set I} \\
                         \phi_{\countvecI} \ema,
                         \right. \right\rangle,\\
      &=\left\langle \bma \alpha \\ \beta \ema
        \;\left|\; \bma \genmat{\set I} \Gamma \\ \bigphi \ema
        \right. \right\rangle,
  \end{aligned}
\end{equation}
where $\beta \in \R^{d_u}$, $\phi_i \in \R^{d_u}$ and
\[
    \bigphi =
      \bma \phi_1 & \phi_2 & \cdots & \phi_{\countvecI} \ema
      \in \R^{d_u \times \countvecI}.
\]
Before proceeding, we note:
\begin{compactitem}
\item The value of $\set J(\set I, \beta, \bigphi)$ is a zonotope in
  $\R^{d_x + d_u}$.
\item The set $\langle \beta \mid \bigphi \rangle$ is a zonotope in
  $\R^{d_u}$, but $\set J(\set I, \beta, \bigphi) \neq \set I \times
  \langle \beta \mid \bigphi \rangle$.
\item Unlike the diagonal matrices $\Gamma$ (and $\Psi$ encountered
  below), the matrix $\bigphi$ is dense: All entries may be nonzero.
\item In the remainder of this work we will often refer to a
  collection of vectors and matrices $\{ \beta(s), \bigphi(s) \}_{s =
    0}^t$.  The fact that the collection is parameterized by time does
  not imply any direct temporal dependence between its elements.
\end{compactitem}

\begin{proposition}
  Given $\{ \beta(s), \bigphi(s) \}_{s = 0}^{t-1}$ and assuming a system
  with no disturbance input or drift, the reach set for an initial
  state space zonotope $\set I$ can be represented by a center,
  generator count and generator matrix given by
  \begin{equation} \label{e:zonotope-no-disturb-evolution}
  \begin{aligned}
    \centvec{\reachSet{t}{\set I}} &= \mat A^t \alpha
      + \sum_{s=0}^{t-1} \mat A^{t-1-s} \mat B \beta(s) \\
    \countvec{\reachSet{t}{\set I}} &= \countvecI \\
    \genmat{\reachSet{t}{\set I}} &= \mat A^t \genmat{\set I} \Gamma
      + \sum_{s=0}^{t-1} \mat A^{t-1-s} \mat B \bigphi(s)
  \end{aligned}
  \end{equation}
\end{proposition}

\begin{proof}
For a system with $\set V = \emptyset$ and $w = 0$, the
dynamics~\refeq{e:linear-dynamics} can be written as the linear
transform
\[
    x(s+1) = \bma \mat A & \mat B \ema \bma x(s) \\ u(s) \ema.
\]
And therefore
\[
  \begin{aligned}
    \bma x(0) \\ u(0) \ema \in \set J&(\set I, \beta(0), \bigphi(0)) \\
    & \;\Leftrightarrow\;
      x(1) \in \bma \mat A & \mat B \ema
        \set J(\set I, \beta(0), \bigphi(0)), \\
    & \;\Leftrightarrow\;
      \reachSet{1}{\set I} =
      \bma \mat A & \mat B \ema \set J(\set I, \beta(0), \bigphi(0)).
  \end{aligned}
\]
Since a linear transformation of a zonotope is a zonotope, we can
apply the linear transformation to~\refeq{e:augmented-zonotope-defn}
to determine
\[
  \begin{aligned}
    \centvec{\reachSet{1}{\set I}} &= \mat A \alpha + \mat B \beta(0), \\
    \countvec{\reachSet{1}{\set I}} &= \countvecI, \\
    \genmat{\reachSet{1}{\set I}} &= \mat A \genmat{\set I} \Gamma
      + \mat B \bigphi(0),
  \end{aligned}.
\]
The result~\refeq{e:zonotope-no-disturb-evolution} can then be derived
by induction on $s = 1, \ldots, T-1$.
\end{proof}

With this evolution formula for the initial zonotope $\set I$, we can
deduce the necessary constraints for viability.

\begin{proposition} \label{t:viable-set-containment-no-disturb}
  Given $\{ \beta(t), \bigphi(t) \}_{t = 0}^{T-1}$ and assuming a system
  with no disturbance input or drift, $\set I \subseteq
  \viabSet{[0,T]}{\set X}$ if
  \begin{equation} \label{e:viable-state-no-disturb-box-containment}
  \begin{aligned}
    \left( \begin{gathered}
      \mat A^t \alpha + \sum_{s=0}^{t-1} \mat A^{t-1-s} \mat B \beta(s) \\      
        - \left| \mat A^t \genmat{\set I} \Gamma
                 + \sum_{s=0}^{t-1} \mat A^{t-1-s} \mat B \bigphi(s)\right|
          \ones{\countvecI}
     \end{gathered} \right)
    &\geq \underline x, \\
    \left( \begin{gathered}
      \mat A^t \alpha + \sum_{s=0}^{t-1} \mat A^{t-1-s} \mat B \beta(s) \\      
        + \left| \mat A^t \genmat{\set I} \Gamma
                 + \sum_{s=0}^{t-1} \mat A^{t-1-s} \mat B \bigphi(s)\right|
          \ones{\countvecI}
     \end{gathered} \right)
     &\leq \overline x
  \end{aligned}
  \end{equation}
  and
  \begin{equation} \label{e:viable-input-no-disturb-box-containment}
  \begin{aligned}
    \beta(t) - \left| \bigphi(t) \right| \ones{\countvecI}
    & \geq \underline u, \\
    \beta(t) + \left| \bigphi(t) \right| \ones{\countvecI}
    & \leq \overline u
  \end{aligned}
  \end{equation}
  for $t = 0, \ldots, T$.
\end{proposition}

\begin{proof}
  Observe that $x(0) \in \set I$ implies $x(t) \in \reachSet{t}{\set
    I}$ by~\refeq{e:reach-defn}.  Define the (not necessarily unique)
  input
  \begin{equation} \label{e:viable-input-no-disturb}
      u(t) = \beta(t) + \bigphi(t) \lambda(x(t), \reachSet{t}{\set I}).
  \end{equation}
  By~\refeq{e:lambda-defn},
  $u(t) \in \langle \beta(t) \mid \bigphi(t) \rangle$.
  By~\refeq{e:control-constraints} and
  Lemma~\ref{t:zonotope-box-containment}, the
  constraint~\refeq{e:viable-input-no-disturb-box-containment} implies $\langle
  \beta(t) \mid \bigphi(t) \rangle \subseteq \set U$.  Therefore,
  $u(t) \in \set U$ for all $t = 0, \ldots, T$.
  By~\refeq{e:state-constraints},~\refeq{e:zonotope-no-disturb-evolution}
  and Lemma~\ref{t:zonotope-box-containment}, the
  constraint~\refeq{e:viable-state-no-disturb-box-containment} implies
  $\reachSet{t}{\set I} \subseteq \set X$ and consequently $x(t) \in
  \set X$.  We have therefore proved for any $x(0) \in \set I$ that
  there exists feasible input signal $u\sig$ such that the trajectory
  $x\sig$ starting from $x(0)$ generated by $u\sig$ satisfies $x(t)
  \in \set X$ for $t = 0, \ldots, T$.  By~\refeq{e:viable-defn}, $\set
  I \subseteq \viabSet{[0,T]}{\set X}$.
\end{proof}

By Proposition~\ref{t:viable-set-containment-no-disturb}, finding a
viable set $\set I$ reduces to finding $\alpha$, $\gamma$ and $\{
\beta(t), \bigphi(t) \}_{t = 0}^{T-1}$ to
satisfy~\refeq{e:viable-state-no-disturb-box-containment}
and~\refeq{e:viable-input-no-disturb-box-containment}.  We can simply
substitute~\refeq{e:viable-state-no-disturb-box-containment}
and~\refeq{e:viable-input-no-disturb-box-containment}
for~\refeq{e:invariant-no-input-box-containment} in the
optimization~\refeq{e:invariant-optimization-no-input} and add $\{
\beta(t), \bigphi(t) \}_{t = 0}^{T-1}$ to the decision variables.  A
feasible solution to the resulting optimization problem will define a
set of viable states through $\alpha$ and $\gamma$, and a (time
dependent) set of viable controls through $\{ \beta(t), \bigphi(t)
\}_{t = 0}^{T-1}$; however, the set of viable controls may be very
small: By~\refeq{e:viable-input-no-disturb} the range of $u(t)$ is
directly proportional to $|\bigphi(t)|$, but reducing $|\bigphi(t)|$
always makes it easier to
satisfy~\refeq{e:viable-input-no-disturb-box-containment} and the
objective from~\refeq{e:invariant-optimization-no-input} provides no
direct incentive to increase $|\bigphi(t)|$ (although a nonzero value
may be necessary to
satisfy~\refeq{e:viable-state-no-disturb-box-containment}).

In order to achieve a larger set of viable controls we would like to
maximize $|\bigphi(t)|$, but maximizing an absolute value is a
non-convex objective and we are not willing to destroy the convexity
of our optimization to directly incorporate such a term.  We also
cannot constrain $\bigphi(t) \geq 0$ elementwise, since the sign of
elements of $\phi_i(t)$ relative to the sign of elements of the
corresponding $\genvec{i}{\set I}$ may be critical to achieving
viability (more discussion in section~\ref{s:viable-example}). Another
approach is needed to seek broader control authority.

%-----------------------------------------------------------
\subsection{Control as Scaled Disturbance}
\label{s:viability-control-disturb}

Intuitively, we would like to characterize the range of control input
which could be applied while maintaining the viability of $\set I$.
That control input authority might cause the reach set to be larger
than it would be otherwise, but such growth may be accommodated in
regions where the state constraints are not tight.  Furthermore, we
already have a method of mathematically characterizing the effect of
such a priori indeterminant input: Treat it as a disturbance.

With that in mind, we introduce a parameterized zonotope
$\controlDisturbSet$ to capture this control input authority which is
independent of the authority in the equations above.  Define
\begin{equation} \label{e:free-input-zonotope-defn}
  \begin{aligned}  
    \controlDisturbSet &= \langle \zeros{d_u}
        \mid \psi_1 \genvec{1}{\controlDisturbSet}, 
             \psi_2 \genvec{2}{\controlDisturbSet},
             \ldots, 
             \psi_{\countvec{\controlDisturbSet}}
             \genvec{\countvec{\controlDisturbSet}}{\controlDisturbSet}
        \rangle, \\
      &=\langle \zeros{d_u} \mid \genmat{\controlDisturbSet} \Psi \rangle,
  \end{aligned}
\end{equation}
where $\Psi \in \R^{\countvec{\controlDisturbSet} \times \countvec{\controlDisturbSet}}$ is
the diagonal matrix with vector $\psi$ along its diagonal.  Like $\set
I$, we will fix the generators $\genmat{\controlDisturbSet} \in \R^{d_u \times
  \countvec{\controlDisturbSet}}$ but leave the generator scaling vector $\psi \in
\R^{\countvec{\controlDisturbSet}}$ with $\psi \geq 0$ as a free parameter.  We
then allow for time-dependent $\controlDisturbSet(t)$ and compute the result of
input $u(t) \in \controlDisturbSet(t)$ on the reach set according to
Proposition~\ref{t:reach-set-evolution-no-control}.
%Although not mathematically necessary, for notational simplicity we
%will assume $\countvec{\controlDisturbSet(t)}$ is the same for all $t$.

\begin{proposition}
  Given $\{ \beta(s), \bigphi(s), \psi(s) \}_{s = 0}^{t-1}$ and
  assuming a system with no disturbance input or drift, the reach set
  for an initial state space zonotope $\set I$ can be represented by a
  center, generator count and generator matrix given by
  \begin{equation} \label{e:zonotope-control-disturb-evolution}
  \begin{aligned}
    \centvec{\reachSet{t}{\set I}} &= \mat A^t \alpha
      + \sum_{s=0}^{t-1} \mat A^{t-1-s} \mat B \beta(s) \\
    \countvec{\reachSet{t}{\set I}} &=
      \countvecI + \sum_{s = 0}^{t-1} \countvec{\controlDisturbSet(s)} \\
    \genmat{\reachSet{t}{\set I}} &=
      \bma \mat F_{\set I} & \mat F_0 & \mat F_1 & \cdots & \mat F_{t-1} \ema
  \end{aligned}
  \end{equation}
  where
  \begin{equation} \label{e:control-disturb-generator-defn}
  \begin{aligned}  
    \mat F_{\set I} &= \mat A^t \genmat{\set I} \Gamma
                     + \sum_{s=0}^{t-1} \mat A^{t-1-s} \mat B \bigphi(s) \\
    \mat F_s &= \mat A^{t-1-s} \mat B \genmat{\controlDisturbSet(s)} \Psi(s)
    \quad \text{for } s = 0, \ldots, t-1
  \end{aligned}
  \end{equation}
\end{proposition}

\begin{proof}
  Combine~\refeq{e:reach-set-zonotope-no-control}
  and~\refeq{e:zonotope-no-disturb-evolution} by superposition.
\end{proof}

We note in passing that the equation for $\centvec{\reachSet{t}{\set
    I}}$ in~\refeq{e:zonotope-control-disturb-evolution} demonstrates
why it is sufficient to fix the center of $\controlDisturbSet(t)$ at
the origin in~\refeq{e:free-input-zonotope-defn}: The parameter
$\beta(t)$ already provides a mechanism to shift the center of the
input set at time $t$.

\begin{proposition} \label{t:viable-set-containment-control-disturb}
  Given $\{ \beta(t), \bigphi(t), \psi(t) \}_{t = 0}^{T-1}$ and
  assuming a system with no disturbance input or drift, $\set I
  \subseteq \viabSet{[0,T]}{\set X}$ if
  \begin{equation} \label{e:viable-state-control-disturb-box-containment}
  \begin{aligned}
    \left( \begin{gathered}
      \mat A^t \alpha + \sum_{s=0}^{t-1} \mat A^{t-1-s} \mat B \beta(s) \\      
        - \left| \mat A^t \genmat{\set I} \Gamma
                 + \sum_{s=0}^{t-1} \mat A^{t-1-s} \mat B \bigphi(s)\right|
          \ones{\countvecI} \\
        - \sum_{s=0}^{t-1} \left(
               \left| \mat A^{t-1-s} \mat B \genmat{\controlDisturbSet(s)} \right|
                           \psi(s) \right)
     \end{gathered} \right)
    &\geq \underline x, \\
    \left( \begin{gathered}
      \mat A^t \alpha + \sum_{s=0}^{t-1} \mat A^{t-1-s} \mat B \beta(s) \\      
        + \left| \mat A^t \genmat{\set I} \Gamma
                 + \sum_{s=0}^{t-1} \mat A^{t-1-s} \mat B \bigphi(s)\right|
          \ones{\countvecI} \\
        + \sum_{s=0}^{t-1} \left(
               \left| \mat A^{t-1-s} \mat B \genmat{\controlDisturbSet(s)} \right|
                           \psi(s) \right)
     \end{gathered} \right)
     &\leq \overline x
  \end{aligned}
  \end{equation}
  and
  \begin{equation} \label{e:viable-input-control-disturb-box-containment}
  \begin{aligned}
    \beta(t) - \left| \bigphi(t) \right| \ones{\countvecI}
      - \left| \genmat{\controlDisturbSet(t)} \right| \psi(t)
      & \geq \underline u, \\
    \beta(t) + \left| \bigphi(t) \right| \ones{\countvecI}
      + \left| \genmat{\controlDisturbSet(t)} \right| \psi(t)
      & \leq \overline u
  \end{aligned}
  \end{equation}
  for $t = 0, \ldots, T$.
\end{proposition}

We have shifted $\Psi(t)$ outside the absolute value
in~\refeq{e:viable-state-control-disturb-box-containment}
and~\refeq{e:viable-input-control-disturb-box-containment} because
$\psi(t) \geq 0$, and then taken advantage of the fact that $\Psi(t)
\ones{\countvec{\controlDisturbSet(t)}} = \psi(t)$.

\begin{proof}
  Let
  \[
    \begin{aligned}
      \tilde{\set U}(t) &= \langle \beta(t) \mid
                         \bma \bigphi(t)
                              & \genmat{\controlDisturbSet(t)} \Psi(t) \ema
                         \rangle, \\
      \tilde{\bigphi}(t) &=
        \bma \bigphi(t)
          & \zeros{d_u \times
                   \sum_{s=0}^{t-1} \countvec{\set \controlDisturbSet(s)}}
        \ema,
    \end{aligned}
  \]
  and then choose any $\rho(t) \in \R^{\countvec{\controlDisturbSet(t)}}$ such that
  \begin{equation} \label{e:rho-defn}
    % Need to surround first component with {} in order to get unary
    % minus spacing correct.
    {-\ones{\countvec{\controlDisturbSet(t)}}} \leq \rho(t) \leq
    +\ones{\countvec{\controlDisturbSet(t)}}.
  \end{equation}
  Define the input
  \begin{equation} \label{e:viable-input-control-disturb}
    u(t) = \beta(t) + \tilde{\bigphi}(t) \lambda(x(t), \reachSet{t}{\set I})
           + \genmat{\controlDisturbSet(t)} \Psi(t) \rho(t).
  \end{equation}
  We augment $\tilde{\bigphi}(t)$ with zero columns / generators to
  account for the extra
  generators~\refeq{e:zonotope-control-disturb-evolution} in
  $\reachSet{t}{\set I}$ arising from $\controlDisturbSet(s)$ for $s <
  t$.  Note that these extra generators have no direct effect on the
  choice of input for step $t$
  in~\refeq{e:viable-input-control-disturb} because they are zero
  vectors, although we do need to account for their continuing effect
  on the choice of $\lambda(x(t), \reachSet{t}{\set I})$ through these
  extra columns in $\tilde{\bigphi}(t)$.
  
  By~\refeq{e:lambda-defn}, the fact that the extra generators in
  $\tilde{\bigphi}(t)$ compared to $\bigphi(t)$ are all zero vectors,
  and~\refeq{e:rho-defn}, $u(t) \in \tilde{\set U}(t)$.
  By~\refeq{e:control-constraints} and
  Lemma~\ref{t:zonotope-box-containment}, the
  constraint~\refeq{e:viable-input-control-disturb-box-containment}
  implies $\tilde{\set U}(t) \subseteq \set U$.  Therefore, $u(t) \in
  \set U$ for all $t = 0, \ldots, T$.
  By~\refeq{e:state-constraints},~\refeq{e:zonotope-control-disturb-evolution}
  and Lemma~\ref{t:zonotope-box-containment}, the
  constraint~\refeq{e:viable-state-control-disturb-box-containment}
  implies $\reachSet{t}{\set I} \subseteq \set X$ and consequently
  $x(t) \in \set X$.  The remainder is the same as the proof of
  Proposition~\ref{t:viable-set-containment-no-disturb}.
\end{proof}

\begin{remark}
  The
  constraints~\refeq{e:viable-state-control-disturb-box-containment}
  and~\refeq{e:viable-input-control-disturb-box-containment} are
  linear in $\alpha$ and $\{ \beta(t), \psi(t) \}_{t = 0}^{T-1}$, and
  are convex in $\gamma$ and $\{ \bigphi(t) \}_{t=0}^{T-1}$.
\end{remark}

Our viability optimization problem is then written as
\begin{equation} \label{e:viable-optimization-control-disturb}
  \begin{aligned}
    \max_{\alpha, \gamma, \{ \beta(t), \bigphi(t), \psi(t) \}_{t = 0}^{T-1} } \;
      & \ones{\countvec{\set I}}^T \gamma
        + \eta \sum_{t=0}^{T-1} \ones{\countvec{\controlDisturbSet(t)}}^T \psi(t) \\
    \text{such that } & \gamma \geq 0 \\
                      & \psi(t) \geq 0 \\   
  \text{and } & \text{\refeq{e:viable-state-control-disturb-box-containment},~\refeq{e:viable-input-control-disturb-box-containment} hold }
                \forall t \in 0, \ldots, T
  \end{aligned}
\end{equation}
where $\eta > 0$ is a weighting parameter used to trade off the
relative importance of large $\gamma$ (to encourage a larger set of
viable states) against large $\psi(t)$ (to encourage a larger set of
viable controls).

\begin{remark}
  The optimization~\refeq{e:viable-optimization-control-disturb} is a
  convex program with $d_x + \countvecI + T d_u(\countvecI + 1) +
  \sum_t \countvec{\controlDisturbSet(t)}$ decision variables, $\countvecI +
  \sum_t \countvec{\controlDisturbSet(t)}$ non-negativity constraints, and $2 (d_x
  + d_u) (T+1)$ general convex constraints
  from~\refeq{e:viable-state-control-disturb-box-containment}
  and~\refeq{e:viable-input-control-disturb-box-containment}.
\end{remark}

%-----------------------------------------------------------
\subsection{Example: Double Integrator} \label{s:viable-example}

\begin{figure}
  \includegraphics[width=0.22\textwidth]{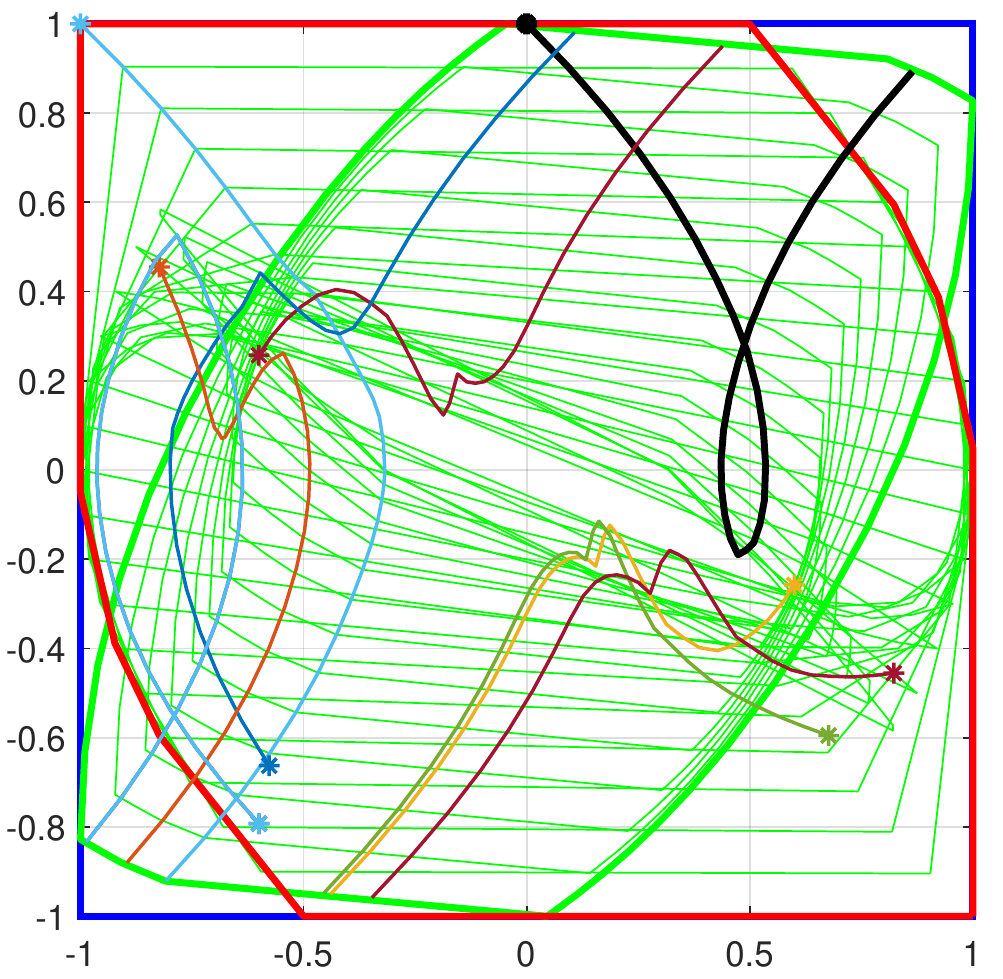}
  \includegraphics[width=0.22\textwidth]{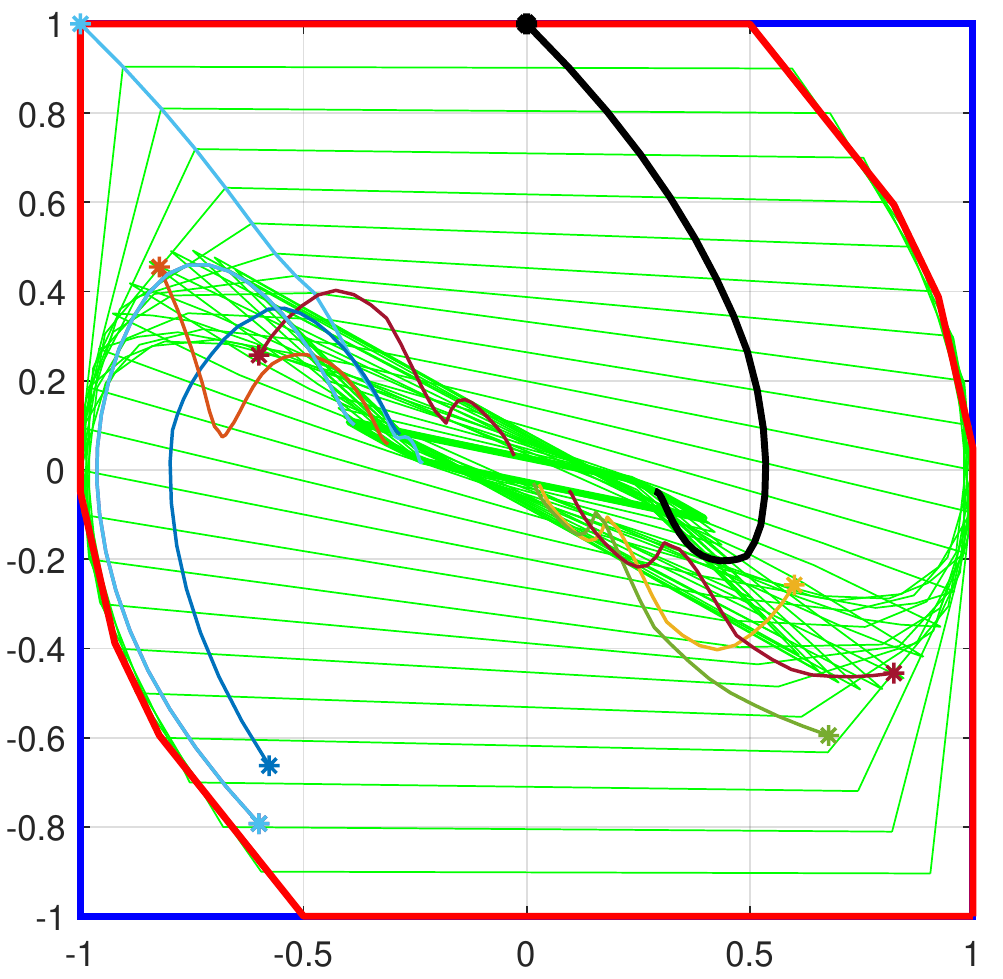}
  \caption{Computed viable set $\set I$ (red thick line) for the
    double integrator example with (left) and without (right)
    $\controlDisturbSet$.  Eight generators equally spaced in the
    north-west quadrant are used (only five have scaling $\gamma_i >
    0.01$).  Also shown are the constraint set $\set X$ (blue thick
    line), $\reachSet{\set I}{t}$ (thin green lines) for $t = 1,
    \ldots T$, and a collection of sample trajectories (star shows the
    initial conditions for each).}
  \label{f:viable-example}
\end{figure}

\begin{figure}
  \includegraphics[width=0.22\textwidth]{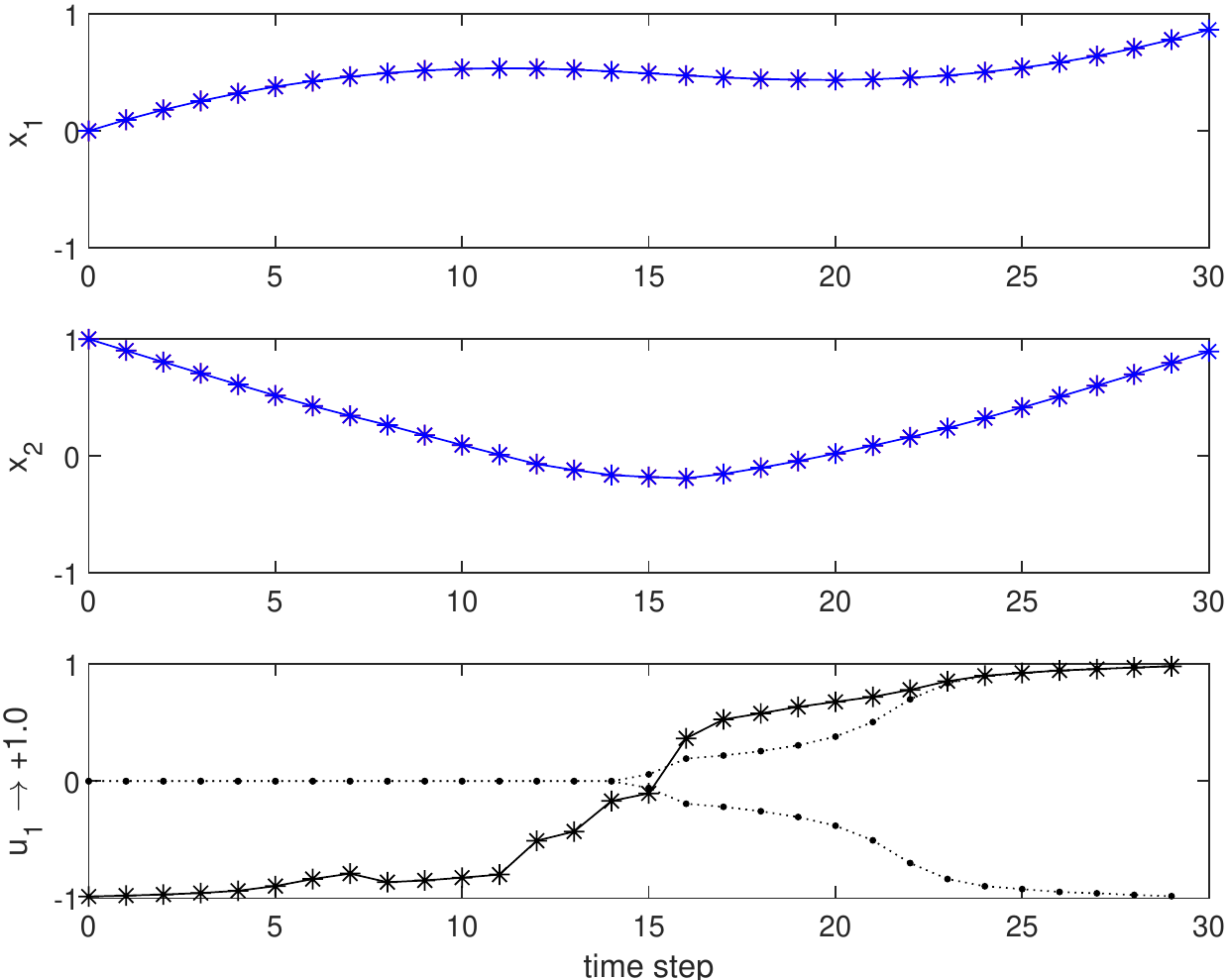}
  \includegraphics[width=0.22\textwidth]{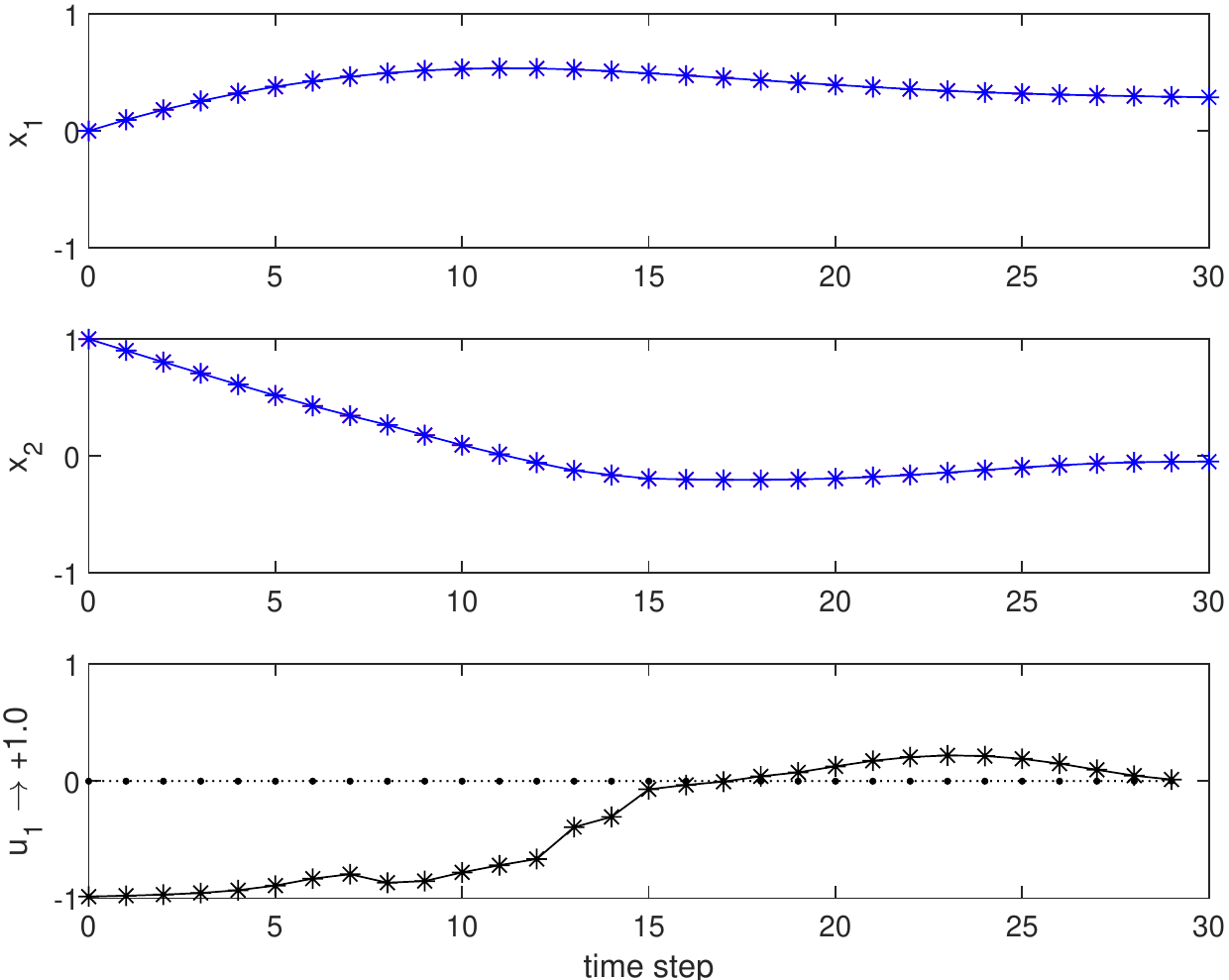}
  \caption{Sample viable trajectory (shown as the thick black
    trajectory in figure~\ref{f:viable-example}) with (left) and without
    (right) $\controlDisturbSet$.  Control at each time is chosen as
    close to $+1$ as possible subject
    to~\refeq{e:viable-input-control-disturb} on the left
    and~\refeq{e:viable-input-no-disturb} on the right.  The dotted
    curve in the control subplot shows the scaling $\psi_1(t)$ (which
    is always zero on the right).  Any difference between the actual
    control $u_1(t)$ and scaling $\psi_1(t)$ arises from $\Phi(t)$.}
  \label{f:viable-example-trajectory}
\end{figure}

To demonstrate the viability algorithm we consider the traditional
double integrator.  Starting from the continuous time system
\[
    \dot x = \bma 0 & +1 \\ 0 & 0 \ema x + \bma 0 \\ 1 \ema u
\]
we use the matrix exponential with time step $0.1$ and \matlab's
\com{integral()} adaptive quadrature routine to construct the discrete
time system
\[
  x(t+1) = \bma +1.0000 & +0.1000 \\ 0 & +1.0000 \ema x(t) +
           \bma +0.0050 \\ +0.1000 \ema u(t).
\]
We use the optimization~\refeq{e:viable-optimization-control-disturb}
with $d_x = 2$, $d_u = 1$, $\underline x = {-\ones{d_x}}$, $\overline
x = +\ones{d_x}$, $\underline u = -1$, $\overline u = +1$, and $T =
30$ to compute a viable set, and in the process demonstrate the
importance of including characterizations of control authority from
both sections~\ref{s:viability-augment-generators}
and~\ref{s:viability-control-disturb}.  Figure~\ref{f:viable-example}
shows the computed viable set with and without the additional control
authority enabled by the technique from
section~\ref{s:viability-control-disturb} (with
$\controlDisturbSet(t)$ equal to the $2 \times 2$ identity matrix for
all $t$).  Run time for the optimization was below 5~seconds for both
cases.

No viable set is found if we omit $\bigphi(t)$ from
section~\ref{s:viability-augment-generators}, so that case is not
shown.  Some intuition for the failure of this latter case to find any
viable set can be found by examining $\phi_1(t)$ for the other cases.
This input component corresponds to the state generator
$\genvec{1}{\set I} = \bma 0 & 1 \ema^T$, which is the generator whose
scaling to a large extent determines the height of the viable set.
For the first half of the time horizon, $\phi_1(t) < 0$; in other
words, states that have a large positive velocity $x_2(t)$ will be
forced (through $\lambda(x(t), \reachSet{t}{\set I})$
in~\refeq{e:viable-input-no-disturb}
or~\refeq{e:viable-input-control-disturb}) to choose a correspondingly
large negative control input $u(t)$.  Without the coupling of state
and input signal through $\bigphi(t)$, viability is infeasible.
Figure~\ref{f:viable-example-trajectory} shows a single trajectory's
states and control input componentwise over time and displays just
this behavior: Although the desired input is $+1$, the early input
signal $u(t)$ is forced to be negative because $x_2(t)$ is quite
positive.

\begin{figure}
  \includegraphics[width=0.22\textwidth]{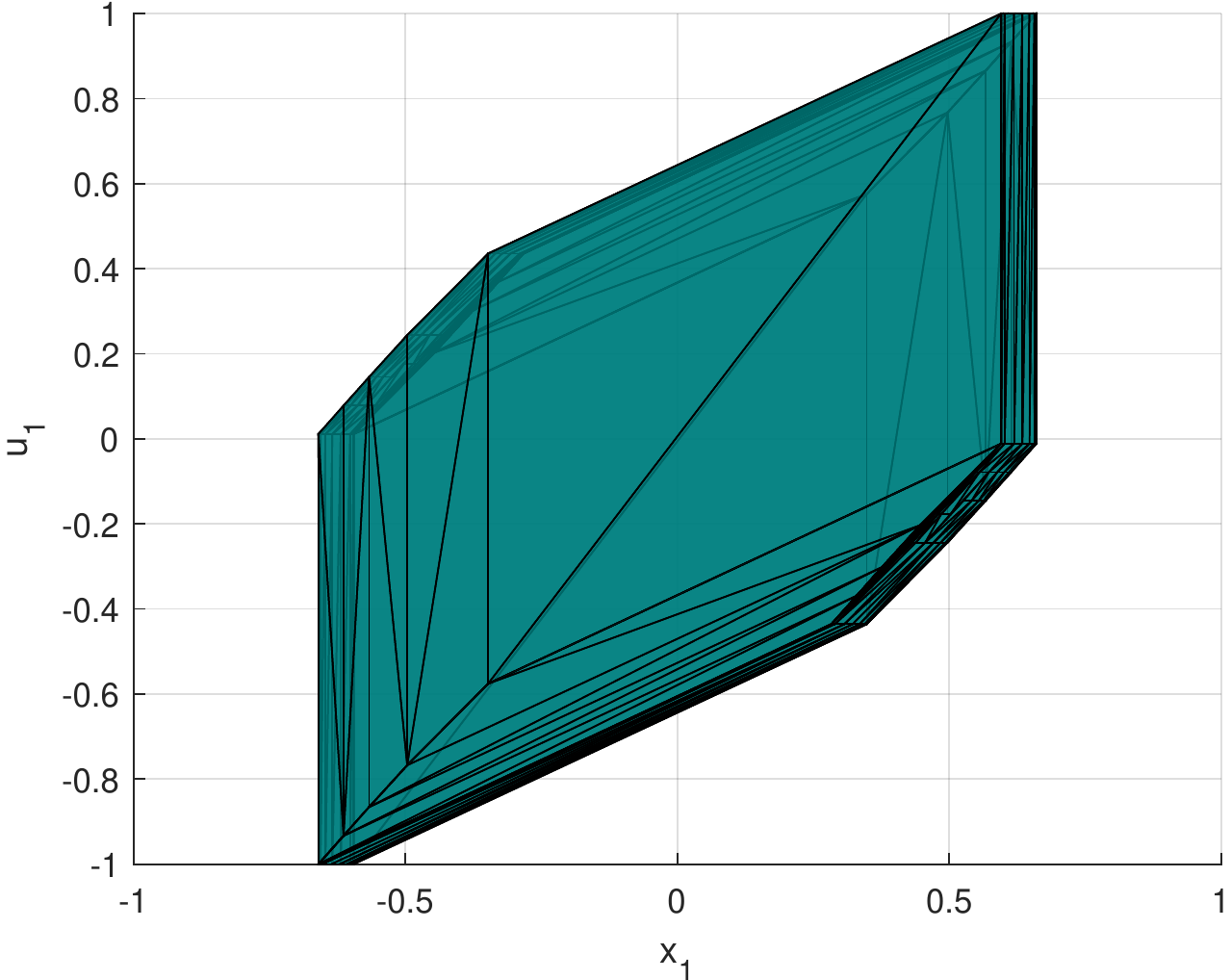}
  \includegraphics[width=0.22\textwidth]{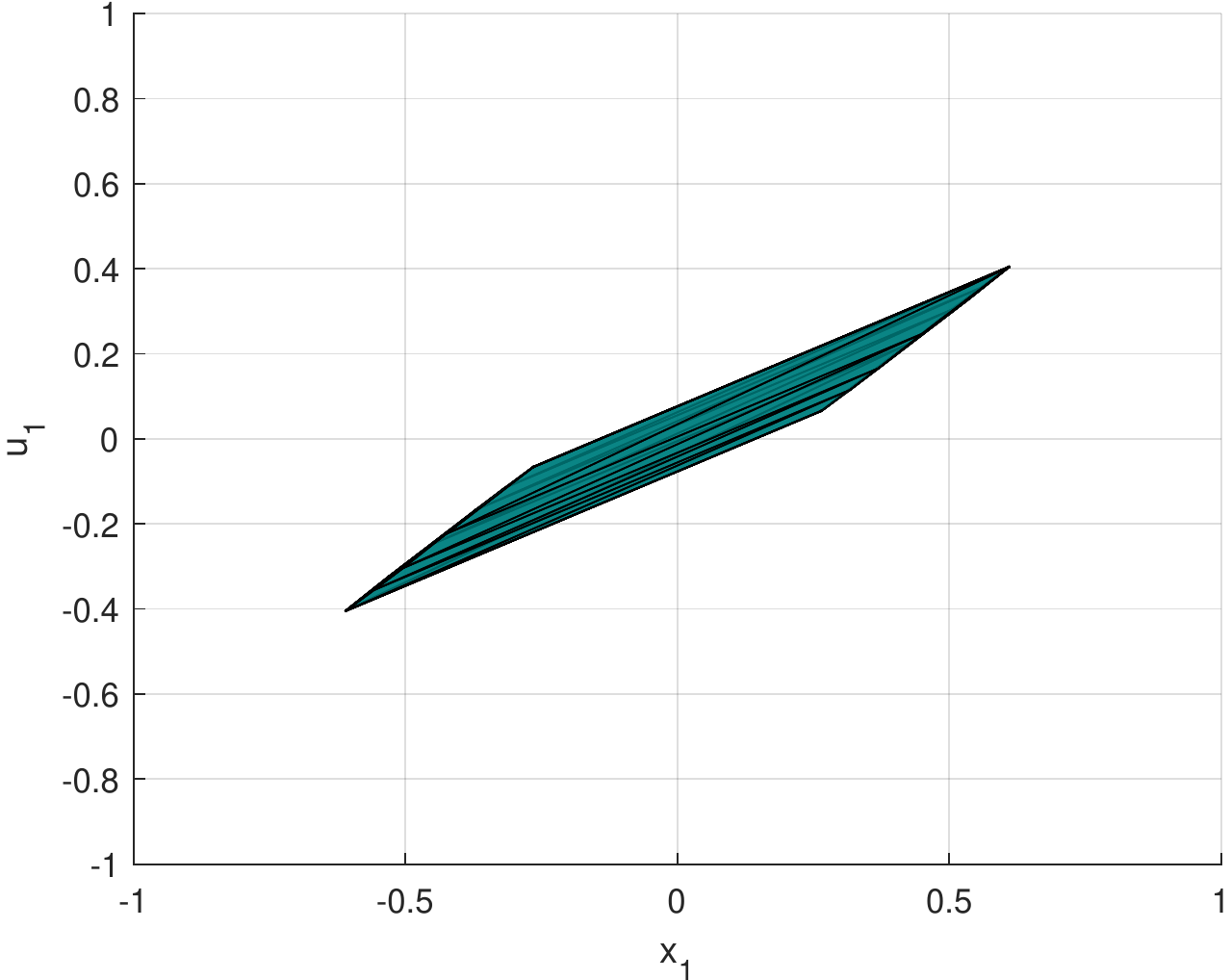}
  \caption{Projection of the set of viable inputs at $t = 20$ with
    (left) and without (right) $\controlDisturbSet$.}
  \label{f:viable-example-controls}
\end{figure}

At later times $\phi_1(t)$ does become slightly positive, but it is at
these times toward the end of the viability horizon that the benefits
of section~\ref{s:viability-control-disturb} are visible.  Although
the sets of viable states in figure~\ref{f:viable-example} are the
same with or without $\controlDisturbSet$, the set of viable controls
is notably larger if we include a non-empty $\controlDisturbSet$; for
example, figure~\ref{f:viable-example-controls} shows that the set of
viable controls at $t = 20$ is much larger on the left, while on the
left of figure~\ref{f:viable-example-trajectory} $u_1(t)$ is much more
able to approach the target value of $+1$ in the latter half of the
horizon because of the growth of $\psi_1(t)$.

%% file: discriminating.tex
%-------------------------------------------------------------
\section{Computing Discriminating Sets} \label{s:discriminating}

In this section we incorporate the control, disturbance and drift
terms by combining the derivations in
Sections~\ref{s:with-disturbance}
and~\ref{s:viability-control-disturb} to handle the full generality of
the dynamics~\refeq{e:linear-dynamics}.  The longer equations
referenced in this section are collected in
figure~\ref{f:long-equations}.

\begin{figure*}
  \centering
  \begin{minipage}{\textwidth}
  \begin{equation} \label{e:zonotope-full-dynamics-evolution}
    \left.    
    \begin{aligned}
    \centvec{\reachSet{t}{\set I}} &= \mat A^t \alpha
      + \sum_{s=0}^{t-1} \mat A^{t-1-s}
        \left( \mat B \beta(s) + \mat C \centvec{\set V} + w \right),\\
    \countvec{\reachSet{t}{\set I}} &= \countvecI
      + \sum_{s = 0}^{t-1} \countvec{\controlDisturbSet(s)}
      + t \countvec{\set V} \\
    \genmat{\reachSet{t}{\set I}} &=
      \bma \mat F_{\set I} & \mat F_0 & \mat F_1 & \cdots & \mat F_{t-1}
             & \mat A^{t-1} \mat C \genmat{\set V}
             & \mat A^{t-2} \mat C \genmat{\set V}
             & \cdots
             & \mat C \genmat{\set V}
      \ema
    \end{aligned}
    \right\}
    \begin{gathered}
      \text{where } \mat F_{\set I} \text{ and } \{ \mat F_s \}_{s=0}^{t-1} \\
      \text{are given in~\refeq{e:control-disturb-generator-defn}}.
    \end{gathered}
  \end{equation}

  \begin{equation} \label{e:discriminating-state-box-containment}
    \left.
    \begin{aligned}
      \left( \begin{gathered}  
      \mat A^t \alpha + \sum_{s=0}^{t-1} \mat A^{t-1-s}
        \left( \mat B \beta(s) + \mat C \centvec{\set V} + w \right)
        - \left| \mat A^t \genmat{\set I} \Gamma
            + \sum_{s=0}^{t-1} \mat A^{t-1-s} \mat B \bigphi(s)\right|
          \ones{\countvecI} \\
        - \sum_{s=0}^{t-1} \left(
            \left| \mat A^{t-1-s} \mat B \genmat{\controlDisturbSet(s)} \right|
                           \psi(s) \right)
        - \sum_{s=0}^{t-1} \left(
            \left| \mat A^{t-1-s} \mat C \genmat{\set V} \right|
            \ones{\countvec{\set V}} \right)
      \end{gathered} \right)
    &\geq \underline x, \\
      \left( \begin{gathered}  
      \mat A^t \alpha + \sum_{s=0}^{t-1} \mat A^{t-1-s}
        \left( \mat B \beta(s) + \mat C \centvec{\set V} + w \right)
        + \left| \mat A^t \genmat{\set I} \Gamma
                 + \sum_{s=0}^{t-1} \mat A^{t-1-s} \mat B \bigphi(s)\right|
          \ones{\countvecI} \\
        + \sum_{s=0}^{t-1} \left(
               \left| \mat A^{t-1-s} \mat B \genmat{\controlDisturbSet(s)} \right|
                           \psi(s) \right)
        + \sum_{s=0}^{t-1} \left(
            \left| \mat A^{t-1-s} \mat C \genmat{\set V} \right|
            \ones{\countvec{\set V}} \right)
      \end{gathered} \right)
     &\leq \overline x
    \end{aligned}
    \right\}
    \text{for } t = 0, \ldots, T.
  \end{equation}
  \end{minipage}
  \caption{Some long equations for
    Propositions~\ref{t:zonotope-full-dynamics-evolution}
    and~\ref{t:discriminating-set-containment}.}
  \label{f:long-equations}
\end{figure*}

\begin{proposition} \label{t:zonotope-full-dynamics-evolution}
  Given $\{ \beta(s), \bigphi(s), \psi(s) \}_{s = 0}^{t-1}$, the reach
  set for an initial state space zonotope $\set I$ can be represented
  by a center, generator count and generator matrix given
  by~\refeq{e:zonotope-full-dynamics-evolution}.
\end{proposition}

\begin{proof}
  Combine~\refeq{e:reach-set-zonotope-no-control}
  and~\refeq{e:zonotope-control-disturb-evolution} by superposition.
\end{proof}

\begin{proposition} \label{t:discriminating-set-containment}
  Given $\{ \beta(t), \bigphi(t), \psi(t) \}_{t = 0}^{T-1}$,
  if~\refeq{e:viable-input-control-disturb-box-containment}
  and~\refeq{e:discriminating-state-box-containment} hold then
  $\set I \subseteq \discSet{[0,T]}{\set X}$.
\end{proposition}

\begin{proof}
  Redefine
  \[
      \tilde{\bigphi}(t) =
        \bma \bigphi(t)
          & \zeros{d_u \times
                   \sum_{s=0}^{t-1} \countvec{\set \controlDisturbSet(s)}}
          & \zeros{d_u \times t \countvec{\set V}}
        \ema,
  \]
  to account for the extra columns / generators in
  $\genmat{\reachSet{t}{\set I}}$ arising from the disturbance inputs,
  and then combine the proofs of
  Propositions~\ref{t:invariant-set-disturb-containment}
  and~\ref{t:viable-set-containment-control-disturb} by superposition.
\end{proof}

%-------------------------------------------------------------
\subsection{Example: Nonlinear Quadrotor} \label{s:quadrotor}

\begin{figure*}
  \centering
  \begin{minipage}{\textwidth}
  \begin{equation} \label{e:dynamics-quad-linear}
  %\small
    \bma \dot x_1 \\ \dot x_2 \\ \dot x_3 \\ 
         \dot x_4 \\ \dot x_5 \\ \dot x_6 \ema
      = \overbrace{
            \bma 0 & 0 & 1 & 0 & 0 & 0 \\
                 0 & 0 & 0 & 1 & 0 & 0 \\
                 0 & 0 & 0 & 0 & +K\bar{u}_1 \cos \bar{x}_5 & 0 \\
                 0 & 0 & 0 & 0 & -K\bar{u}_1 \sin \bar{x}_5 & 0 \\
                 0 & 0 & 0 & 0 & 0 & 1 \\
                 0 & 0 & 0 & 0 & -d_0 & -d_1 \ema}^{\mat A}
          \overbrace{\bma x_1 \\ x_2 \\ x_3 \\ x_4 \\ x_5 \\ x_6 \ema }^{x}  
          +
          \overbrace{
            \bma 0 & 0 \\
                 0 & 0 \\
                 K \sin \bar{x}_5 & 0 \\
                 K \cos \bar{x}_5 & 0 \\
                 0 & 0 \\
                 0 & n_0 \ema}^{\mat B}
          \overbrace{\bma u_1 \\ u_2 \ema}^{u}
        + \overbrace{
            \bma 0 & 0 \\ 0 & 0 \\ 1 & 0 \\ 0 & 1 \\ 0 & 0 \\ 0 & 0 \ema
            }^{\mat C}
          \overbrace{\bma v_1 \\ v_2 \ema }^{v}
        + \overbrace{
          \bma 0 \\ 0 \\ 
               -K \bar{u}_1 \sin \bar{x}_5 \\
               +K \bar{u}_1 \cos \bar{x}_5 - g \\ 
               0 \\ 0 \ema
          }^{w}
  \end{equation}
  \end{minipage}
  \caption{Linearized quadrotor dynamics.}
\end{figure*}

To demonstrate the utility of the full algorithm, we compute a
discriminating set for a partially nonlinear six dimensional
longitudinal model of a quadrotor taken from~\cite{bouffard12}.  The
state space dimensions are:
\begin{compactitem}
  \item horizontal position $x_1$ [m] (positive rightward),
  \item vertical position $x_2$ [m] (positive upward),
  \item horizontal velocity $x_3$ [m/s],
  \item vertical velocity $x_4$ [m/s],
  \item roll $x_5$ [rad] (positive clockwise),
  \item roll velocity $x_6$ [rad/s].\\
\end{compactitem}
The control input dimensions are:
\begin{compactitem}
  \item total thrust $u_1$,
  \item desired roll angle $u_2$.\\
\end{compactitem}
The nonlinear continuous time plant dynamics model is:
\begin{subequations}
  \label{e:dynamics-quad-nonlinear}
  \begin{align}
    \dot x_1 &= x_3, \\
    \dot x_2 &= x_4, \\
    \dot x_3 &= u_1 K \sin x_5, \\
      \label{e:dynamics-quad-nonlinear-x3}
    \dot x_4 &= -g + u_1 K \cos x_5, \\
      \label{e:dynamics-quad-nonlinear-x4}
    \dot x_5 &= x_6, \\
    \dot x_6 &= -d_0 x_5 - d_1 x_6 + n_0 u_2, 
      \label{e:dynamics-quad-nonlinear-x6}
  \end{align}
\end{subequations}
We adopt the constraint set used in~\cite{MYLT16}, except that we
broaden the range of $x_5$:
\begin{equation} \label{e:quad-state-constraints}
  \begin{aligned}
    x_1 &\in [ -1.7, +1.7 ], \\
    x_2 &\in [ +0.3, +2.0 ], \\
    x_3 &\in [ -0.8, +0.8 ], \\
    x_4 &\in [ -1.0, +1.0 ], \\
    x_5 &\in [ -\tfrac{\pi}{12}, +\tfrac{\pi}{12} ], \\
    x_6 &\in [ -\tfrac{\pi}{2}, +\tfrac{\pi}{2} ].
  \end{aligned}
\end{equation}
We also broaden the range of allowed controls:
\begin{equation} \label{e:quad-input-constraints}
  \begin{aligned}
    u_1 &\in [ -1.5, +1.5 ] + \bar{u}_1, \\
    u_2 &\in [ -\tfrac{\pi}{12}, +\tfrac{\pi}{12} ] + \bar{x}_5,
  \end{aligned}
\end{equation}
where $\bar u_1 = g / K$ and $\bar x_5 = 0$.  We then
linearize~\refeq{e:dynamics-quad-nonlinear} about $\bar u_1$ and $\bar
x_5$ to arrive at the linear model~\refeq{e:dynamics-quad-linear}.
Note that unlike~\cite{MYLT16} we do not hybridize the dynamics in
order to linearize about multiple operating points; the entire
constraint set is handled with a single linear model.
For the range of $x_5$ in~\refeq{e:quad-state-constraints} and $u_1$
in~\refeq{e:quad-input-constraints}, the linearization errors are
\[
  \begin{aligned}
    \text{error in } \dot x_3 &\in [ -0.2760, +0.2760 ], \\
    \text{error in } \dot x_4 &\in [ 0.0000, +0.3668 ].
  \end{aligned}
\]
For conservativeness, we choose $\set V$ to be this error rectangle
dilated by $10\%$.  Note that the disturbance affects only $x_3$ and
$x_4$ because the dynamics of the remaining dimensions are exactly
linear.  The parameter values used in simulation are taken
from~\cite{MYLT16}: $K=0.89/1.4$, $d_{0}=70$, $d_{1}=17$ and
$n_{0}=55$.

We use \matlab's \com{expm()} and \com{integral()} routines to
construct the discrete time version of~\refeq{e:dynamics-quad-linear}
with time step $0.05$, and then compute for a horizon $T = 40$.  To
choose generators, we note that for $\bar x_5 = 0$ the pairs of
dimensions $\{ (x_1, x_3), (x_2, x_4), (x_3, x_5), (x_5, x_6) \}$ look
like double integrators.  For these pairs we create five generators in
the north-west quadrant, while for the remaining pairs of states we
create only two generators along the diagonal and anti-diagonal.  To
these we add the six coordinate axes for a total of $\countvec{\set I}
= 48$ generators.  We also tried running the optimization with an
additional 52~randomly oriented generators to see whether coupling
between more than two dimensions would improve the discriminating set;
however, none of the randomly oriented generators achieved a scaling
$\gamma_i \geq 0.01$ and so we discarded them in subsequent runs.

\begin{figure}
  \includegraphics[width=0.45\textwidth]{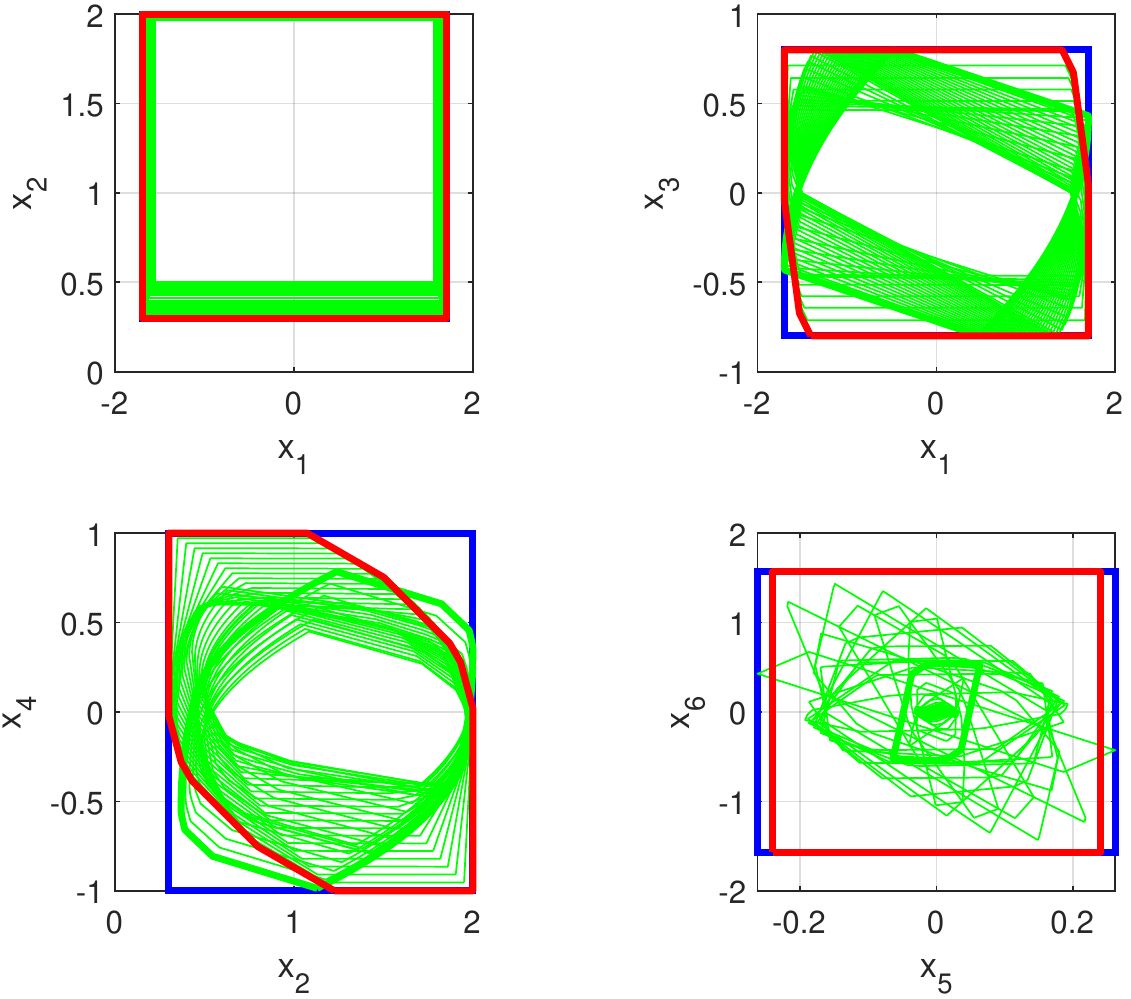}
  \caption{Projections of the computed discriminating set $\set I$
    (red thick line) for the quadrotor model.}
  \label{f:quad-discriminating}
\end{figure}

We run the optimization
problem~\refeq{e:viable-optimization-control-disturb} with
constraint~\refeq{e:discriminating-state-box-containment} substituted
for~\refeq{e:viable-state-control-disturb-box-containment}.  Run time
for the optimization is just over 5~minutes.  Only 11 of the 48
generators achieved a scaling factor $\gamma_i \geq 0.01$.
Figure~\ref{f:quad-discriminating} shows projections of the resulting
discriminating set.  Note that this set is discriminating for both the
linear and nonlinear models, since we conservatively capture all
linearization error in the disturbance input bounds.

Although we have insufficient space to make a detailed comparison with
the ellipsoid-based results from~\cite{MYLT16}, we can observe that in
roughly the same computational time (albeit on a slightly faster
laptop) our zonotope-based algorithm is able to find a much larger
discriminating set over twice the time horizon with double the time
resolution.  Furthermore, the zonotope representation is much less
conservative with its treatment of the disturbance input, and hence we
do not need to hybridize or to restrict the range of $x_5$ and $u_1$
so severely.  We hypothesize that most of the improvement in accuracy
arises from the fact that zonotopes can exactly represent the
rectangular form of typical state and input constraints, while
ellipsoids are forced to adopt dramatic under- or over-approximations.

%% file: conclusion.tex
%-------------------------------------------------------------
\section{Discussion} \label{s:discussion}

The formulations developed above leave a number of parameters to be
chosen by the user and make some assumptions which could be relaxed.
We briefly discuss these issues in this section.

The key parameter open to the user is the choice of generators
$\genmat{\set I}$ and $\genmat{\set F(t)}$.  Unless there is some
reason to believe that direct coupling of the inputs is beneficial,
the latter will typically be chosen as an identity matrix.  Choosing
the former is trickier; however, the experience in
section~\ref{s:quadrotor} indicates that examination of the sparsity
pattern of $\mat A$ may allow one to choose these vectors more
efficiently than simply trying to cover the unit hypersphere in
$\R^{d_x}$.

One factor that is perhaps not so obvious when choosing generators is
that this choice impacts the quality of the resulting sets not just
directly through the generators but also indirectly through the
heuristic objective function
in~\refeq{e:viable-optimization-control-disturb}.  If $\genmat{\set
  I}$ is orthonormal then $\max \ones{\countvec{\set I}}^T \gamma$ is
not an unreasonable heuristic to make $\set I$ large, but as the
generators lose perpendicularity (inevitable as the number of
generators grows) the quality of this heuristic decreases.  It should
be noted that this heuristic appears to encourage sparse $\gamma$,
which could be a significant benefit for downstream uses of $\set I$.

Finally, the restriction to interval hulls of the control input set
$\set U$ and constraint set $\set X$ was driven by the need to include
constraints in the optimization which confirmed that (projections of)
zonotopes were contained within those sets.  That restriction can be
relaxed to any class of sets for which a reasonable number of
constraints can confirm zonotope containment; for example,
intersections of slabs, or even convex polygons with a modest number
of faces.  It may even be possible to allow full zonotopes
using~\cite[Lemma~3]{HREA16}, albeit at the cost of swapping a small
number of constraints (such
as~\refeq{e:viable-input-control-disturb-box-containment}) for a full
linear matrix inequality.

%-------------------------------------------------------------
\section{Conclusions and Future Work} \label{s:conclusion}

We have derived convex optimizations whose solutions represent
invariant, viable or discriminating sets for discrete time, continuous
state affine dynamical systems, and demonstrated the results on a
simple rotation, a double integrator and a six dimensional nonlinear
longitudinal model of a quadrotor respectively.  The optimizations can
be solved with modest computing power in seconds to minutes.

A key shortcoming of the current formulation is its restriction to
discrete time.  Unfortunately, the typical approaches to soundly
mapping continuous time reachability into discrete time reachability
(for example, see~\cite{AF16}) cannot be used in our approach, so we
are exploring alternatives.

We are hopeful that some of the computational efficiency techniques
demonstrated in~\cite{bakduggirala17a,Bogomolov+18} could be applied
to improve the scalability of the optimization problems derived here.
Although there are pruning heuristics which could be applied to reduce
the number of generators, the size of the optimization inevitably
grows with finer time discretization and/or longer horizons.  We have
considered approaches to replace one large optimization with many
smaller ones, although there are tradeoffs in such schemes.

Finally, we are exploring the use of the viable and discriminating
sets for classifying and filtering exogenous control inputs at
run-time.